\documentclass[12pt,a4paper]{article}

\usepackage{authblk}

\usepackage{a4wide}

\usepackage{color}
\usepackage{tikz}
\usepackage{textcomp}

\usepackage{algorithm}
\usepackage{algpseudocode}

\usepackage{a4wide}
\usepackage{textcomp}

\usepackage[english]{babel}

\usepackage{amsmath, amssymb, amsfonts}
\usepackage{fancyvrb}
\usepackage{graphicx}

\usepackage[pagestyles]{titlesec}
\newpagestyle{copyright}{
\setfoot[][][\footnotesize\copyright~B.~Wouters (UvA), 2019 \hfill \mbox{}] 
{}{}{\hfill \mbox{} \footnotesize\copyright~B.~Wouters (UvA), 2019}
}

\numberwithin{equation}{section}

\usepackage{transparent}

\usepackage{centernot}

\usepackage{diagbox}

\usepackage{bm}

\newcommand{\dd}[1]{\,\text{d}{#1}}

\newcommand{\Ex}[1]{\mathrm{E}\!\left[{#1}\right]}

\newcommand{\prob}[1]{\text{P}\!\left[{#1}\right]}
\newcommand{\Var}[1]{\text{Var}\!\left[{#1}\right]}

\usepackage{color}

\VerbatimFootnotes

\usepackage[hyperfootnotes=false, hidelinks]{hyperref}

\newcommand{\bu}{{\bf u}}
\newcommand{\bv}{{\bf v}}
\newcommand{\bx}{{\bf x}}
\newcommand{\by}{{\bf y}}
\newcommand{\bz}{{\bf z}}

\newcommand{\bvarepsilon}{{\bm{\varepsilon}}}
\newcommand{\btheta}{{\bm{\theta}}}
\newcommand{\bphi}{{\bm{\phi}}}

\newcommand{\bmu}{{\bm{\mu}}}

\newcommand{\bxi}{{\bm{\xi}}}

\newcommand{\bchi}{{\bm{\chi}}}

\usepackage{mathtools}
\DeclarePairedDelimiter{\norm}{\lVert}{\rVert}

\newcommand{\byt}{{\bf y}_t}
\newcommand{\bxt}{{\bf x}_t}
\newcommand{\bepst}{\bm{\varepsilon}_t}
\newcommand{\bepspara}{\bm{\varepsilon}^\parallel}
\newcommand{\bepstpara}{\bm{\varepsilon}_t^\parallel}
\newcommand{\bepsperp}{\bm{\varepsilon}^\perp}
\newcommand{\bepsperphat}{\hat{\bm{\varepsilon}}^\perp}
\newcommand{\bepstperp}{\bm{\varepsilon}_t^\perp}
\newcommand{\bepstc}{\bm{\varepsilon}_t^\mathrm{c}}

\newcommand{\bxit}{\bm{\xi}_t}
\newcommand{\bchit}{\bm{\chi}_t}
\newcommand{\bxiepst}{\bxit^{(\varepsilon)}}
\newcommand{\bzeta}{\bm{\zeta}}
\newcommand{\bzetat}{\bm{\zeta}_t}

\newcommand{\bepstnull}{\bm{\varepsilon}_t^\mathrm{\bf o}}
\newcommand{\bet}{\mathrm{\bf e}_t}

\newcommand{\bybar}{\bar{\bf y}}

\newcommand{\bxtden}{{\bf x}_t^\mathrm{den}}
\newcommand{\bxtopt}{{\bf x}_t^\mathrm{opt}}

\newcommand{\bxthatden}{\hat{\bf x}_t^\mathrm{den}}
\newcommand{\bxthatopt}{\hat{\bf x}_t^\mathrm{opt}}

\newcommand{\Sigmay}{\Sigma_\by}
\newcommand{\Sigmayhat}{\hat{\Sigma}_\by}
\newcommand{\Sigmaylagged}[1]{\Sigma_\by({#1})}
\newcommand{\Sigmaylaggedhat}[1]{\widehat{\Sigma}_\by({#1})}
\newcommand{\Sigmaylaggedtilde}[1]{\tilde{\Sigma}_\by({#1})}
\newcommand{\Sigmax}{\Sigma_\bx}
\newcommand{\Sigmaxlagged}[1]{\Sigma_\bx({#1})}
\newcommand{\Sigmaeps}{\Sigma_\bvarepsilon}
\newcommand{\Sigmaepspara}{\Sigma_{\bepspara}}
\newcommand{\Sigmaepsperp}{\Sigma_{\bepsperp}}
\newcommand{\Sigmaepsperphat}{\widehat{\Sigma}_{\bepsperp}}
\newcommand{\Sigmaepslagged}[1]{\Sigma_\bvarepsilon({#1})}
\newcommand{\Sigmaxi}{\Sigma_\bxi}
\newcommand{\Sigmaxilagged}[1]{\Sigma_\bxi({#1})}
\newcommand{\Sigmachi}{\Sigma_\bchi}
\newcommand{\Sigmachihat}{\hat{\Sigma}_\bchi}
\newcommand{\Sigmaxieps}{\Sigma_{\bxi^\varepsilon}}

\newcommand{\Vhat}{\hat{V}}
\newcommand{\Khat}{\widehat{K}}
\newcommand{\Wperphat}{\widehat{W}_\perp}

\newcommand{\mse}{\textbf{MSE}}
\newcommand{\Tr}[1]{\text{Tr}\!\left[{#1}\right]}

\newcommand{\Pden}{P_\mathrm{den}}
\newcommand{\Popt}{P_\mathrm{opt}}
\newcommand{\Portho}{P_\mathrm{ortho}}

\newcommand{\Popthat}{\hat{P}_\mathrm{opt}}

\newcommand{\normsq}[1]{\norm{ {#1} }^2}

\newcommand{\bigO}[1]{O \left( {#1} \right)}
\newcommand{\bigOp}[1]{O_\mathrm{P} \! \left( {#1} \right)}

\newcommand{\tracevar}[1]{\mathrm{Tr} \! \left[ \mathrm{Var} \! \left[ {#1} \right] \right]}

\usepackage{amsthm}
\newtheorem{theorem}{Theorem}
\newtheorem{assumption}{Assumption}
\newtheorem{lemma}{Lemma}

\newcommand{\bZ}{{\bf Z}}

\newcommand{\bZt}{{\bf Z}_t}

\usepackage[round,comma,authoryear]{natbib}

\title{Model-agnostic noise reduction for high-dimensional time series data}

\author{Bram Wouters}
\author{Cees Diks}
\affil{University of Amsterdam}

\begin{document}

\maketitle

\begin{abstract}
\noindent 
We develop a model-agnostic framework for noise reduction in high-dimensional time series that explicitly targets optimal recovery of a low-dimensional latent dynamic component contaminated by observational white noise. Under the assumption that the latent dynamics live in a low-dimensional linear dynamic subspace,
we characterize the optimal linear projection onto the dynamic subspace and provide a geometric description of the residual error in terms of the relative orientation of the signal and noise spaces. We propose estimators for the dynamic subspace and the optimal projection based on lagged covariance matrices, bootstrap dimension selection, and a low-rank representation of the structured noise. Under mild conditions, the resulting denoised series is shown to converge to its population target at the usual parametric rate. Simulations show that the proposed method can substantially improve subspace estimation, reconstruction error, and one-step-ahead forecast accuracy compared with both orthogonal projection-based denoising and the raw data. The approach is illustrated by empirical applications to high-dimensional stock returns and to a 20-variate time series of macroeconomic indicators. Code is available at \url{https://github.com/brwo/NoiseReductionForTimeSeries}.
\end{abstract}

\newpage

\newpage
\section{Introduction} \label{sec:introduction}

High-dimensional time series are ubiquitous in statistics and econometrics. Examples include panels of macroeconomic indicators, large portfolios of financial assets, spatially indexed climate or environmental fields, and collections of biomedical or sensor signals. In such settings, it is often natural to view the observed process as a noisy manifestation of a lower-dimensional dynamic system: a small number of latent components drive most of the serial dependence, while the remaining variation reflects measurement error, idiosyncratic shocks, or other forms of high-dimensional noise. Accurately extracting this latent dynamic signal is essential for forecasting, structural interpretation, and policy analysis.

Dynamic factor models provide a prominent framework for capturing this structure. In macroeconomics, large approximate factor models have become standard for summarizing many predictors with a few diffusion indexes and for improving forecasts; see, for instance, \citet{Stock2002,Stock2002a} and \citet{Bai2002,Bai2007}. In a complementary spectral approach, generalized dynamic factor models based on the eigenstructure of the spectral density matrix have been developed by \citet{Forni2000} and further refined in \citet{Forni2005}. \citet{Hallin2007} proposed information criteria to determine the number of common shocks in such models. For high-dimensional vector time series, \citet{Pan2008} and \citet{Lam2011,Lam2012} showed how to estimate a low-dimensional dynamic factor space using lagged autocovariances in a way that is robust to additive white noise. Related ideas have also been used to study the number of factors in very high-dimensional settings, for example in \citet{Xia2018}.

Similar concepts arise in functional time series, where each observation is a function (or curve) rather than a finite-dimensional vector. \citet{Bathia2010} studied the problem of identifying a finite-dimensional dynamic subspace underlying a curve time series and proposed bootstrap-based tests for its dimension. Functional factor approaches have been successfully used for forecasting age-specific mortality and fertility rates in the presence of measurement errors and outliers \citep{Hyndman2007}. These developments underscore the broad appeal of viewing complex data as lying near a low-dimensional dynamic structure embedded in a noisy high-dimensional space.

At the same time, there is a substantial and largely parallel literature on denoising and regularization in high dimensions. In the context of independent observations, wavelet-based signal denoising pioneered by \citet{Donoho1994,Donoho1995} showed that sparse representations can yield near-optimal risk properties.  In high-dimensional covariance estimation, shrinkage and factor-based estimators have been proposed to stabilize eigenvalues and improve conditioning; see, for example, \citet{Ledoit2004} and \citet{Fan2011,Fan2013}. These methods highlight the importance of exploiting structure, either in the underlying signal or in the covariance matrix, to reduce noise effectively in high dimensions.

Most of the dynamic factor literature has focused on two main tasks: determining the dimension and orientation of the low-dimensional dynamic space, and using the associated factors for forecasting or structural analysis. The noise is typically treated as an idiosyncratic remainder that is uncorrelated over time and relatively small compared with the common component. In practice, however, the noise can be large, structured, and partially aligned with the signal directions. Orthogonally projecting the data onto an estimated factor space therefore does not necessarily yield an optimally denoised series. By contrast, many denoising and covariance regularization methods are not tailored to exploit the specific combination of temporal dependence and cross-sectional factor structure that arises in high-dimensional time series.

The aim of this paper is to develop and analyze a noise reduction method for high-dimensional time series that explicitly targets optimal recovery of the low-dimensional dynamic component, while remaining agnostic about the detailed parametric form of its dynamics. We assume only that there exists a low-dimensional linear ``dynamic subspace'' that carries all serial dependence, and that the noise is white but may have arbitrary cross-sectional covariance and lives in a linear subspace that may (partly) overlap with the dynamic space. Within this general framework, we address three questions: how to define an optimal linear denoiser that preserves the dynamic subspace; how the geometry between the dynamic and noise subspaces determines what can be achieved in principle; and how such an optimal denoiser can be estimated from data and used in practice.

\paragraph{Our contributions}
Our contributions are threefold.
First, at the population level we formalize noise reduction as the problem of finding the linear projection onto the dynamic space that minimizes the mean-squared reconstruction error of the latent dynamic component. Within the class of linear projections, we characterize a population-optimal denoiser that combines projection onto the dynamic space with a correction that exploits the cross-sectional structure of the noise. A geometric analysis links achievable denoising to the relative geometry of the dynamic and noise spaces and identifies when perfect linear denoising is possible.

Second, we turn the population construction into an estimation method for a single observed high-dimensional time series. Building on lagged autocovariances \citep{Lam2012} and bootstrap-based dimension selection \citep{Bathia2010}, we estimate the dynamic space and then the structured noise from residual covariances, using an explained-variance rule for the noise dimension. The method uses only second-order properties and, seen as a pre-processing step, can be paired with many downstream models on the denoised series.

Third, we show that, under standard conditions, the denoised series obtained from the estimated optimal transformation converges to its population target at the usual square-root rate. Simulations with vector autoregressions and mixed noise show that optimal denoising markedly improves parameter estimates, reconstruction error, and one-step-ahead forecasting performance relative to projection-based denoising and raw data. 

Overall, the proposed method can be viewed as a bridge between the dynamic factor and denoising literatures: it uses lagged autocovariance structure to learn the low-dimensional dynamic subspace and the cross-sectional noise structure to design a linear transformation that reduces noise as effectively as possible within that space.

\paragraph{Related literature}

Our work is related to several strands of literature on high-dimensional and functional time series, dynamic factor models, and high-dimensional denoising.

The first strand concerns factor modeling for high-dimensional time series and panels. For large cross-sectional datasets, \citet{Bai2002,Bai2007} developed influential information criteria for determining the number of factors in approximate factor models. In macroeconomic forecasting, \citet{Stock2002,Stock2002a} showed that diffusion indexes or principal components constructed from large panels of macro indicators can substantially improve forecasts. Spectral approaches to dynamic factor modeling have been developed by \citet{Forni2000,Forni2005}, who introduced and analyzed generalized dynamic factor models based on the spectral density matrix. \citet{Hallin2007} proposed criteria to determine the number of common shocks in the general dynamic factor model. For multivariate time series, \citet{Pan2008} and \citet{Lam2011,Lam2012} studied high-dimensional factor models in which a small number of dynamic factors drive the serial dependence, and proposed methods to estimate the factor space from lagged autocovariances in the presence of additive noise. More recent work such as \citet{Xia2018} examines the eigenvalue behavior of high-dimensional covariance structures to infer the number of factors, providing complementary tools for dimension selection.

A second strand is the functional time series literature, where each observation is a curve or function. \citet{Bathia2010} introduced a dynamic dimension-reduction framework for curve time series and proposed bootstrap procedures to determine the finite-dimensional dynamic subspace, exploiting the fact that white noise has vanishing lagged covariance. This connects closely to functional data approaches to forecasting, such as \citet{Hyndman2007}, which treat age-specific mortality and fertility rates as functional data observed with noise. There is also work on regularized and sparse principal component methods in high dimensions, for example \citet{Johnstone2009}, providing insight into the behavior of principal components when the dimension is large relative to the sample size. Our setting can be viewed as a finite-dimensional analog of these frameworks, with an emphasis on exploiting the noise structure for improved denoising.

A third related area is high-dimensional denoising and covariance regularization. Wave\-let-based denoising methods, pioneered by \citet{Donoho1994,Donoho1995}, have demonstrated how combining orthogonal transforms with thresholding can yield near-optimal risk properties in nonparametric regression and signal processing. In the context of high-dimensional covariance estimation, \citet{Ledoit2004} proposed a well-conditioned shrinkage estimator that stabilizes eigenvalues, while \citet{Fan2011,Fan2013} developed factor-based and thresholding methods for estimating large covariance matrices in approximate factor models. Related work on high-dimensional regression and regularization, such as the Dantzig selector of \citet{Candes2007}, highlights the importance of exploiting low-dimensional structure in high-dimensional problems more generally. Our approach is complementary to these methods; instead of focusing on static covariance or regression problems, we target noise reduction in a dynamic setting where both temporal dependence and cross-sectional factor structure are present, and we design a denoiser that is explicitly optimal within a natural class of linear transformations.


\section{Methodology} \label{sec:methodology}
We consider an observable vector-valued time series $\{\byt\}$, with $\byt \in\mathbb{R}^n,$ and $t \in \mathbb{Z}$,  composed of a latent signal component $\bxt \in\mathbb{R}^n$ and a latent white-noise component $\bepst \in\mathbb{R}^n,$
\begin{equation*} \label{eq:time_series_definition}
\byt = \bxt + \bepst 
\end{equation*}
The signal component contains, by definition, all serial dependence of the observed time series. We assume it can be represented as a $d$-dimensional factor model, $\bxt = U \bxit,$ where $U$ is an $(n \times d)$ factor-loading matrix and $\bxit \in \mathbb{R}^d$ are the factors. We assume both $\{\bxit\}$ and $\{\bepst\}$ are strictly stationary time series processes with finite second moments, so they are also weakly stationary. Denoting the covariance matrix of any weakly stationary vector-valued time series $\{\bv_t\}$ by $\Sigma_\bv = \Var{\bv_t},$ we assume the factor time series $\{\bxit\}$ has a full-rank covariance matrix $\Sigmaxi,$ ruling out the possibility of redundancy in the factor model. Other details of the underlying process defining the signal are left unspecified. The noise component $\bepst$ is white noise with mean zero and covariance matrix  $\Sigmaeps.$

The dimension of the factor model is typically small and much smaller than the dimension of the observed time series, i.e., $d \ll n.$ The method proposed here does not necessarily need this assumption, but in typical applications only then the dimension reduction is useful.

Without loss of generality, we assume that the columns of $U$ are orthonormal and the covariance matrix of the factors, $\Sigmaxi,$ is assumed to be diagonal. Starting from a factor model $\bxt = U \bxit$ that does not have these properties, we can use the QR-decomposition $U=QR$ and redefine $U$ as $Q$ and $\bxit$ as $R \bxit.$ This makes the columns of $U$ orthonormal. There is still freedom in defining $U,$ since for any orthogonal $(d\times d)$-matrix $A$ we have $U \bxit = U A ^\top \bxit$. By redefining $U$ as $UA$ and $\bxit$ as $A^\top \bxit,$ one can choose $A$ such that $\Sigmaxi$ is diagonal. 

We refer to the $d$-dimensional subspace in which $\bxt$ lives, spanned by the columns of $U,$ as the dynamic space $\mathcal{M}.$ The serial dependence of the observed time series exists entirely within this space. Because $\Sigmaxi$ is full-rank, the covariance matrix of the signal, $\Sigmax = U \Sigmaxi U^\top,$ has exactly $d$ non-zero eigenvalues. The corresponding eigenvectors form a basis of the dynamic space $\mathcal{M}.$ We use this fact to estimate the dimension and obtain a basis of $\mathcal{M}$, but $\Sigmax$ cannot be estimated directly because $\bxt$ is unobservable. Moreover, the standard (naive) estimator of the covariance matrix of the observed time series $\{\byt\}$ is biased due to the noise component, since $\Sigmay = \Sigmax + \Sigmaeps.$ This bias does not vanish as the sample size of the observed time series tends to infinity.

\cite{Lam2011} present a method, developed earlier for functional time series \citep{Bathia2010}, that overcomes this problem by exploiting the lagged autocovariance, which is defined for any weakly stationary vector-valued time series $\{\bv_t\}$ as $\Sigma_\bv(k) = \text{Cov} \left[ \bv_{t+k}, \bv_t \right].$ Since the lagged autocovariance of the noise component is zero, $\Sigmaepslagged{k} = 0$ for $k \neq 0,$ the lagged autocovariances of the observed time series and the signal component are equal, $\Sigmaylagged{k} = \Sigmaxlagged{k}$ for $k \neq 0.$ Hence, the observed time series can be used to estimate the autocovariance of $\bxt$ at non-zero lags.

A vector in the orthogonal complement of the dynamic space $\mathcal{M}$ is necessarily an eigenvector of $\Sigmaylagged{k} = U \Sigmaxilagged{k} U^\top$ with eigenvalue zero. Assuming that $\Sigmaxilagged{k}$ is full rank, the kernel of $\Sigmaylagged{k}$ is of dimension $n-d$ and therefore its eigenvectors with eigenvalue zero span the orthogonal complement of $\mathcal{M}.$ As a consequence of the same full-rank assumption, for any $\bu \in \mathbb{R}^n$ it holds that $\Sigmaylagged{k} \Sigmaylagged{k}^\top \bu = 0$ if and only if $\Sigmaylagged{k}^\top \bu = 0.$ Hence, $\mathcal{M}^\perp$ is spanned by the eigenvectors of $\Sigmaylagged{k} \Sigmaylagged{k}^\top$ with eigenvalue zero. Note that $\Sigmaylagged{k} \Sigmaylagged{k}^\top$ is positive semi-definite and, in particular, symmetric. Therefore, its eigenvectors with non-zero eigenvalue span $\mathcal{M}.$

More generally, introducing the matrix
\begin{equation} \label{eq:K}
K = \sum_{k=1}^{k_0} c_k \Sigmaylagged{k} \Sigmaylagged{k}^\top,
\end{equation}
where $c_k \in \mathbb{R}$ are arbitrary coefficients, the eigenvectors of $K$ with non-zero eigenvalue span the dynamic space $\mathcal{M},$ provided that there is at least one $k$ with $c_k \neq 0$ for which the corresponding matrix $\Sigmaxilagged{k}$ is full rank. Crucially, using the naive estimator for $\Sigmaylagged{k},$ one can define a consistent estimator for $K$ without asymptotic bias arising from the noise component. This means that a basis for $\mathcal{M}$ can be estimated without asymptotic bias, using an eigenanalysis of the sample version of the matrix $K.$ Because of the sum in the definition of $K,$ information from different lags is combined. For the default option $c_k = 1$ for $k=1,\ldots,k_0,$ it has been reported that the estimation of the dynamic space is rather insensitive to the choice of $k_0$ (provided $k_0 \geq 2$), both in the case of high-dimensional \citep{Lam2011} and functional \citep{Bathia2010} time series data.

For future reference, we denote by $V$ the $(n \times d)$ matrix whose columns are $d$ orthonormal eigenvectors of $K$ associated with its non-zero eigenvalues. We emphasize that $U$ and $V$ are alike in the sense that their columns form orthonormal bases of the same dynamic space $\mathcal{M}$, but $U$ and $V$ are typically not equal. The matrix $U$ is an inherent property of the time series process under consideration, whereas $V$ also depends on the (partly arbitrary) choice of the coefficients $c_k$ in the definition of $K$ in Equation~\eqref{eq:K}.

\subsection{Noise reduction in time series data} \label{sec:denoising}
Given an observed time series $\{\byt\}_{t=1}^T,$ the goal of denoising is to reconstruct the signal component $\bxt.$ Formally, denoising is a linear operation $\Pden : \mathbb{R}^n \to \mathbb{R}^n$ such that $\bxtden = \Pden \byt$ is the reconstruction of $\bxt.$ In the context of an underlying factor model as described above, the true signal component $\bxt$ lies within the subspace $\mathcal{M} \subset \mathbb{R}^n.$ It therefore makes sense to restrict ourselves to denoising operations that are projections onto the dynamic space $\mathcal{M}.$ This means that $\Pden^2 = \Pden$ and if $\bu \in \mathcal{M},$ then $\Pden \bu = \bu.$ In addition, we limit our analysis to linear operations and will refer to $\Pden$ as a denoising matrix.

An example of such a denoising operation is an orthogonal projection onto $\mathcal{M},$ i.e., $\Portho = U U^\top.$ Although not framed as noise reduction, this is in essence what is proposed in \cite{Lam2011} (and in \cite{Bathia2010} for functional time series data), where an estimate of the dynamic space $\mathcal{M}$  can be used to reconstruct the latent component of the time series described by a factor model. This approach 
does not take into account the possibility of noise taking values in a subspace that is not orthogonal to $\mathcal{M}.$ If there exists such structure in the noise, orthogonal denoising removes noise in a suboptimal way. In what follows, we develop an optimal form of noise reduction that explicitly exploits this structure.

\subsection{MSE-optimal denoising} \label{sec:optimal_denoising}
To address the question of optimal denoising, we first specify the concept of optimality. We say that a denoising matrix $\Pden$ is optimal if it minimizes the mean squared (denoising) error $\Ex{\norm{\bxt - \bxtden}^2},$ where $\norm{\cdot}$ denotes the Euclidean norm. We will use the terms ``MSE-optimal denoising'' and ``optimal denoising'' interchangeably. 

Before presenting the optimal denoising strategy, some notation regarding the noise component $\bepst$ must be introduced, since our strategy leverages structure in the noise. Given a dynamic space $\mathcal{M},$ the noise can be decomposed as $\bepst = \bepstpara + \bepstperp,$ where $\bepstpara = U U^\top \bepst$ and $ \bepstperp=(I_n - U U^\top)\bepst$. The noise is thus decomposed into a geometrically parallel and an orthogonal component with respect to $\mathcal{M}.$
Note that orthogonality in the outcome space $\mathbb{R}^n$ does not imply that $\bepstperp$ and $\bepstpara$ are uncorrelated; in fact, our denoising strategy is based on exploiting their covariance.

We then write the orthogonal noise component in terms of a lower-dimensional representation, $\bepstperp = W_\perp \bchit,$ such that $\Sigmachi$ is full rank. Let $d_\perp$ be the dimension of this representation, i.e., $W_\perp$ is an $(n\times d_\perp)$-matrix and $\bchit \in \mathbb{R}^{d_\perp}.$ With the same reasoning as previously used for $U$ and $\Sigmaxi,$ we can assume without loss of generality that the columns of $W_\perp$ are orthonormal and $\Sigmachi$ is diagonal. Note that this representation can be viewed as a factor model for the noise component, although it should be stressed that there is no serial dependence within the noise. Furthermore, the dimension of the orthogonal noise is bounded from above, $d_\perp \leq n - d.$ In the case of equality, the (linear) space of the orthogonal noise component corresponds to the orthogonal complement of $\mathcal{M}.$ We also emphasize that $d_\perp$ is not assumed to be small relative to $n$ or any other scale. The sole purpose of this lower-dimensional representation is its non-redundancy, making $\Sigmachi$ full rank.

We are now in a position to present the optimal denoising strategy, which is the main result of this paper.

\begin{theorem} \label{thm:optimal_denoising}
Given a dynamic space $\mathcal{M},$ the MSE-optimal denoising matrix, i.e.\ the minimizer of $\Ex{\norm{\bxt - \bxtden}^2}$, among all projections onto $\mathcal{M}$ is given by
\begin{equation} \label{eq:P_optimal}
\Popt = U U^\top \left( I_n - \Sigmay W_\perp ( \Sigmachi )^{-1} W_\perp^\top \right),
\end{equation}
which is the solution of a convex optimization problem.
\end{theorem}

\begin{proof}
The matrix $\Pden$ is meant to act on the time series observations $\byt = \bxt + \bepstpara + \bepstperp,$ which can be written as a linear combination of the columns of $U$ and $W_\perp.$ Since we have restricted our analysis to denoising matrices that are projections onto the dynamic space $\mathcal{M},$ we can, without loss of generality, write $\Pden = U U^\top + U Q W_\perp^\top$. The search for an optimal $\Pden$ is therefore reduced to the search for an optimal $(d \times d_\perp)$-matrix $Q.$ Since $\bxtden = \Pden \byt = \bxt + \bepstpara + U Q W_\perp^\top \bepstperp,$ the mean squared (denoising) error is
\begin{align*}
\Ex{\norm{\bxt - \bxtden}^2}
& = \Tr{ \Var{ \bepstpara + U Q W_\perp^\top \bepstperp }} \\
& = \Tr{\Sigmaepspara} + 2 \Tr{Q W_\perp^\top \Sigmay U} + \Tr{Q^\top Q \Sigmachi} ,
\end{align*}
where we used that $\bepst$ is assumed to have zero mean and that $\text{Cov} \left[ \bepstpara, \bepstperp \right] = U U^\top \Sigmay (I_n - U U^\top).$

This is a convex optimization problem for $Q$, because the Hessian is block diagonal with blocks $\Sigmachi$ that are all positive definite (recall that $\Sigmachi$ is full rank). The first-order condition is $U^\top \Sigmay W_\perp + Q \Sigmachi = 0,$ whose solution is $Q_\mathrm{opt} = - U^\top \Sigmay W_\perp \Sigmachi^{-1}.$ Substituting this back into $\Pden$ gives the desired $P_\mathrm{opt}.$
\end{proof}

The intuition behind this optimal denoising strategy is as follows. Given a dynamic space $\mathcal{M},$ it is possible to use the observed time series $\{\byt\}$ to recover the noise component $\bepstperp$ that is orthogonal to $\mathcal{M}.$ Subtracting this from $\byt$ is what we refer to as orthogonal denoising. This is generally suboptimal, as the parallel noise component $\bepstpara$ is left untouched. By exploiting the structure within the noise space, information about the orthogonal noise component can be used to remove (part of) the parallel noise component, thereby minimizing the mean squared denoising error. To be more concrete, one can view the procedure as performing a multiple linear regression with $\bepstperp$ as the regressor and using its covariance with $\bepstpara$ to also remove (part of) $\bepstpara$ from the observed time series $\{\byt\}$.

To understand how much of the noise can be removed, we write the total noise in the factor-model form $\bepst = W \bxiepst,$ where $W$ is an $(n\times d_\varepsilon)$-matrix and $\Sigmaxieps$ is full rank. Again, without loss of generality, the columns of $W$ are assumed to be orthonormal. They span what we will call the noise space $\mathcal{M}_\varepsilon$. Note that this representation always exists, since no assumptions about the (relative) size of the noise dimension $d_\varepsilon$ are made. 

\begin{theorem} \label{thm:remaining_noise}
Given a dynamic space $\mathcal{M}$ and a noise space $\mathcal{M}_\varepsilon,$ the denoising error of optimal denoising $\bxtopt = \Popt \byt,$ as defined in Equation \eqref{eq:P_optimal}, is given by 
\begin{equation} \label{eq:denoising_error}
    \bxtopt - \bxt = \left\{
        \begin{matrix}
            {\bf 0} & \qquad \mathrm{if} \quad \mathcal{M} \cap \mathcal{M}_\varepsilon = \emptyset, \\
            & \nonumber \\
            \bepstc - \mathrm{Cov} \left[ \bepstc, \bepstperp \right] W_\perp ( \Sigmachi )^{-1} W_\perp^\top \bepstperp  & \qquad \mathrm{if} \quad \mathcal{M} \cap \mathcal{M}_\varepsilon \neq \emptyset,
        \end{matrix}
    \right.
\end{equation}
where $\bepstc = W_\mathrm{c} W_\mathrm{c}^\top \bepst$ and the columns of $W_\mathrm{c}$ form an orthonormal basis of the common subspace $\mathcal{M} \cap \mathcal{M}_\varepsilon$ (the {\normalfont c} stands for ``common'').
\end{theorem}

\begin{proof}
First consider the case $\mathcal{M} \cap \mathcal{M}_\varepsilon = \emptyset.$ The noise space $\mathcal{M}_\varepsilon$ is spanned by the columns of $W,$ where $\bepst = W \bxiepst.$ If the dynamic and noise space have no common subspace, every linear combination of the columns of $W$ has a component orthogonal to $\mathcal{M}.$ Noise $\bepst$ therefore always has an orthogonal component $\bepstperp$ and a (potential) parallel component $\bepstpara$ can always be expressed as a linear function of the perpendicular component. More concretely, if $\mathcal{M} \cap \mathcal{M}_\varepsilon = \emptyset,$ there exists a matrix $M_{\parallel\perp} \in \mathbb{R}^{n\times n}$ such that $\bepst = (M_{\parallel\perp} + I_n)\bepstperp,$ where $I_n$ is the identity matrix on $\mathbb{R}^n.$ This implies
\[
\mathrm{Cov} \! \left[ \bepstpara, \bepstperp \right] = M_{\parallel\perp} \Sigmaepsperp = M_{\parallel\perp} W_\perp \Sigmachi W_\perp^\top
\]
and therefore
\[
U U^\top \Sigmay W_\perp ( \Sigmachi )^{-1} W_\perp^\top \bepstperp = U U^\top \mathrm{Cov} \left[ \bepstpara, \bepstperp \right] W_\perp ( \Sigmachi )^{-1} W_\perp^\top \bepstperp =  M_{\parallel\perp} \bepstperp = \bepstpara.
\]
The optimal denoising matrix $\Popt,$ defined in Equation \eqref{eq:P_optimal}, then projects the noise onto the null vector, i.e., $\Popt \bepst = {\bf 0},$ leading to perfect denoising.

In the case $\mathcal{M} \cap \mathcal{M}_\varepsilon \neq \emptyset$ parallel noise $\bepstpara$ can contain a component in $\mathcal{M} \cap \mathcal{M}_\varepsilon,$ which is given by $\bepstc = W_\mathrm{c} W_\mathrm{c}^\top \bepst,$ where the columns of $W_\mathrm{c}$ form an orthonormal basis of the common subspace $\mathcal{M} \cap \mathcal{M}_\varepsilon$ (the c stands for `common'). The space spanned by the remaining noise parallel to $\mathcal{M},$ $\bepstpara - \bepstc,$ has no common component with the dynamic space $\mathcal{M}$ and thus for this part of the noise the previous situation applies. This means that there exists a matrix $M_{\parallel\perp} \in \mathbb{R}^{n\times n}$ such that $\bepst = \bepstc + (M_{\parallel\perp} + I_n)\bepstperp.$ The optimally denoised signal then simplifies to
\begin{align*}
    \bxtopt 
    & = \Popt \byt \\
    & = \Popt (\bxt + \bepstc + (M_{\parallel\perp} + I_n)\bepstperp) \\
    & = \bxt + \bepstc + (\bepstpara - \bepstc) - U U^\top \Sigmay W_\perp ( \Sigmachi )^{-1} W_\perp^\top \bepstperp \\
    & = \bxt + \bepstc - \mathrm{Cov} \left[ \bepstc, \bepstperp \right] W_\perp ( \Sigmachi )^{-1} W_\perp^\top \bepstperp,
\end{align*}
where in the last equality we used that
\[
U^\top \Sigmay W_\perp = U^\top \mathrm{Cov} \left[ \bepstc, \bepstperp \right] W_\perp + U^\top \mathrm{Cov} \left[ \bepstpara - \bepstc, \bepstperp \right] W_\perp
\]
and $\mathrm{Cov} \left[ \bepstpara - \bepstc, \bepstperp \right] = M_{\parallel\perp} \Sigmaepsperp = M_{\parallel\perp} W_\perp \Sigmachi W_\perp^\top.$ As in the previous case, the component $\bepstpara - \bepstc$ is perfectly denoised. 
\end{proof}

In the previous discussion of optimal denoising and in Theorems \ref{thm:optimal_denoising} and \ref{thm:remaining_noise}, two types of structure within the noise space have been used: linear correlations between $\bepstperp$ and $\bepstpara,$ and a factor-model form of the noise and the orthogonal component of the noise,  $\bepst = W \bxiepst$ and $\bepstperp = W_\perp \bchit$, respectively. The first type of structure is essential for optimal denoising. In the absence of such a structure, i.e., when the covariance between $\bepstperp$ and $\bepstpara$ is zero, Theorem \ref{thm:optimal_denoising} is still valid, but optimal denoising reduces to orthogonal denoising.

By contrast, it is important to note that the factor-model form of the noise is not a restrictive assumption, because we do not assume the dimensions $d_\varepsilon$ and $d_\perp$ to be small. Typical noise $\bepst \thicksim \mathrm{IID} \! \left( {\bm 0}, \Sigmaeps \right)$ fits into the framework $\bepst = W \bxiepst,$ where $\Sigmaeps = W \Sigmaxieps W^\top$ is the eigendecomposition with the zero eigenvalues removed. The same representation applies to $\bepstperp = W_\perp \bchit.$ Assuming that $d_\varepsilon$ and $d_\perp$ are small would be unrealistic, because, unless the noise is subject to context-specific restrictions, there is no reason for the noise to be confined to a lower-dimensional subspace. This suggests that, typically, $d_\varepsilon = n$ and $d_\perp = n - d,$ in which cases $\Sigmaeps$ and $\Sigmaepsperp,$ respectively, are invertible. This does not imply, though, that the invertibility of $\Sigmaeps$ and $\Sigmaepsperp$ is guaranteed. The primary purpose of using the factor-model form for the noise is to ensure that the results in Theorems \ref{thm:optimal_denoising} and \ref{thm:remaining_noise} remain valid even when the noise covariance matrices are not invertible.

It should be noted though that at the sample level optimal denoising is more effective when (most of) the noise can be explained by a small number of factors. This implies that $\Sigmaeps$ and $\Sigmaepsperp$ have a few relatively large eigenvalues, while the remaining eigenvalues are (close to) zero. This in turn makes the (effective) dimensions $d_\varepsilon$ and $d_\perp$ small relative to $n$ and $n-d,$ respectively. In the simulations in Section \ref{sec:simulation} we investigate how the presence of noise that lacks the two types of structure discussed here affects optimal denoising in finite samples.

\subsection{Estimation} \label{sec:estimation}
This section translates the theoretical optimal denoising result of Equation \eqref{eq:P_optimal} into an actionable approach for the practical case of an observed time series sample $\left\{ \byt \right\}_{t=1}^T$ of length $T.$ The optimal denoising matrix contains several quantities that need to be estimated. Regarding the dynamic space $\mathcal{M},$ both the dimension $d$ and a set of basis vectors must be estimated. For this we largely follow the approach of \cite{Lam2011}, as outlined at the start of Section \ref{sec:methodology}. Recall that we made the distinction between the $(n \times d)$-matrices $U$ and $V.$ Both have columns that form an orthonormal basis of $\mathcal{M}.$ In particular, $U U^\top = V V^\top$ and therefore $U$ can be replaced by $V$ in the optimal denoising matrix of Equation \eqref{eq:P_optimal}. A crucial difference is that the columns of $V$ are (by definition) eigenvectors of $K,$ defined in Equation \eqref{eq:K}. This makes it possible to define an estimator of $V,$ in contrast to $U.$ The sample version of the matrix $K$ is accessible by means of the naive estimator for the lagged autocovariance matrix,
\[
\Sigmaylaggedhat{k} = \frac{1}{T - k} \sum_{t=1}^{T-k} (\by_{t+k} - \bybar) (\byt - \bybar)^\top,
\]
where $\bybar$ is the sample mean of the observed time series.

Recall that $K$ has exactly $d$ non-zero (positive) eigenvalues, with corresponding eigenvectors that form a basis of $\mathcal{M}.$ This will typically not be the case for the estimate $\widehat{K},$ due to finite-sample estimation error coming from the noise components $\bepst.$ This complicates the estimation of $d.$ The literature contains many proposals to address this issue. Early approaches try to minimize forecast errors \citep{Hyndman2007} or use information criteria \citep{Bai2002, Bai2007, Hallin2007, Otto2024}. One approach iteratively adds directions to the orthogonal complement of $\mathcal{M},$ until the next addition is no longer white noise \citep{Pan2008}. This method, based on the Ljung-Box-Pierce portmanteau test for white noise, has been reported to be problematic in the current context \citep{Bathia2010}. A widely-used alternative is the ratio-based approach \citep{Lam2011, Lam2012}, which estimates $d$ based on the largest relative drop in eigenvalues of $\widehat{K}.$ Our simulations have shown that, for relatively low signal-to-noise ratios and small time series length $T,$ the largest relative drop can occur at $d + d_\bvarepsilon,$ where $d_\bvarepsilon$ is the dimension of the noise space $\mathcal{M}_\varepsilon.$ This can inadvertently identify the noise components as signal and thereby (largely) overestimates the dimension of $\mathcal{M}.$ A more recent approach by \citet{Wu2018} uses the fact that the asymptotic behavior of eigenvalue estimates differs between non-zero and zero eigenvalues, see also \citet{Otto2024, Otto2025}. This method requires the tuning of a hyperparameter to separate the two types of asymptotic behavior. In a supervised setting this can be accomplished with cross-validation, for example, but we do not observe the would-be target variable $\bxt.$ One option is to tune this hyperparameter indirectly, for example by using the denoised signal for forecasting and taking the resulting forecast error as the tuning criterion.

In this paper, we estimate $d$ using a series of bootstrap tests. This approach was first proposed by \cite{Bathia2010} in the context of functional time series. As it uses neither the smoothness of the curves, nor the infinite dimensionality of the Hilbert space of functional data, this method can also be applied to vector-valued time series. It is worth noting that a satisfying denoising performance often does not require an exactly correct estimate of $d.$ The signal component $\bxt$ can be viewed as a linear combination of the different factors in the factor model. The expected importance of each of the factors, which can be quantified by the corresponding eigenvalue of $\Sigmaxi,$ typically varies among the different factors. This means that a subset of the $d$ factors can account for most, or almost all, of the signal component. Our simulations in Section \ref{sec:simulation} show that $d$ is often underestimated, without much harm to the denoising performance. For details about the bootstrap approach  to estimating $d$ and its effect on the noise reduction, see Appendix \ref{app:bootstrap}. Given an estimated dimension $\hat{d}$ of the dynamic space, we define $\Vhat$ with on its columns the orthonormal eigenvectors of $\Khat$ associated with the $\hat{d}$ largest eigenvalues.

Since $\Sigmay$ can be estimated by $\Sigmaylaggedhat{0},$ what remains to be estimated in Equation \eqref{eq:P_optimal} is an orthonormal basis of the orthogonal noise components, i.e., the columns of $W_\perp,$ and the covariance structure of the corresponding noise factors $\bchit.$ The covariance of the orthogonal noise is estimated via $\Sigmaepsperphat = (I_n - \Vhat \Vhat^\top) \Sigmayhat (I_n - \Vhat \Vhat^\top).$ Given an estimate $\hat{d}_\perp$ of $d_\perp,$ the estimate $\Wperphat$ of $W_\perp$ has, in its columns, the orthonormal eigenvectors of $\Sigmaepsperphat$ associated with the $\hat{d}_\perp$ largest eigenvalues. The corresponding eigenvalues form the diagonal of $\Sigmachihat,$ whose off-diagonal elements are zero. We resolve the problem of estimating $d_\perp$ differently than we did for estimating $d.$ It is possible to use the bootstrap procedure of \citet{Bathia2010} again, if one uses the reconstructions $\bepsperphat_t = (I_n - \Vhat \Vhat^\top) \byt$ as the ``observed'' data 
from which bootstrap samples are generated. This is, however,
sensitive to estimation error in $\Vhat.$ We opt for a more pragmatic approach and choose $\hat{d}_\perp$ such that the fraction of the explained variance of $\bepsperphat_t$ included in $\Wperphat$ exceeds a threshold $\tau_\perp \in (0,1),$
\[
\hat{d}_\perp = \min \left\{ 
d \in \{1,\ldots,n\}
\middle|
\sum_{i=1}^d \hat{\lambda}_i^{(\bepsperp)} \geq \tau_\perp \sum_{j=1}^n \hat{\lambda}_j^{(\bepsperp)}
\right\} ,
\]
where $\hat{\lambda}_i^{(\bepsperp)}$ are the eigenvalues of $\Sigmaepsperphat$ in descending order. The threshold is set by the practitioner. This pragmatic approach suffices, because we are not interested in $d_\perp$ itself. Instead, we use $\bepstperp$ as a regressor for $\bepstpara.$ {\it A priori}, there is no reason why relatively small factors of $\bepstperp$ solely predict large factors of $\bepstpara.$ The distinction between orthogonal and parallel noise is made by the dynamic space $\mathcal{M},$ whose orientation is typically unrelated to the directions of the noise factors. Neglecting small factors in $\bepstperp$ therefore entails little risk of missing a big denoising opportunity. In addition, estimating small factors in $\bepstperp$ is more sensitive to finite-sample estimation error. Neglecting them therefore reduces the risk of overfitting on noise.

\subsection{Asymptotic theory} \label{sec:asymptotics}
The expressions for the denoising errors under optimal denoising in Theorem \ref{thm:remaining_noise} are asymptotic. They can only be achieved with full knowledge of the underlying distribution or with an infinite amount of data available. Considering an observed time series sample $\left\{ \byt \right\}_{t=1}^T,$ a natural question is how fast the asymptotic denoising error is approached as a function of the time series length $T.$ 

The asymptotics of estimating the dynamic space $\mathcal{M}$ have been studied by \cite{Lam2011} and \cite{Lam2012} for finite-dimensional time series, and in \citet{Bathia2010} for functional time series. In both cases, the asymptotics of estimating a basis for $\mathcal{M}$ are derived under the assumption that the number of basis vectors, the dimension $d,$ is known. The fact that the methods used for estimating $d$ are consistent supports this assumption. We will do the same and assume that both $d$ and $d_\perp$ are known. 

Before stating the asymptotics of optimal denoising, we first present the main assumptions that were used in the derivation of that result. More technical assumptions are deferred to Appendix \ref{app:asymptotics}.

\begin{assumption} \label{ass:full_rank_xi}
    No linear combination of the components of $\{\bxit\}$ is white noise.
\end{assumption}
\begin{assumption} \label{ass:stochastic_indepdence}
    The random variables $\bxt$ and $\bvarepsilon_{t'}$ are stochastically independent $\forall t, t'$.
\end{assumption}
\begin{assumption} \label{ass:full_rank_K}
    Considering the definition of $K$ in Equation \eqref{eq:K}, there is at least one $k \in \{1, 2,\ldots,k_0\}$ for which $c_k \neq 0$ and $\Sigmaxilagged{k}$ is full-rank.
\end{assumption}
\begin{assumption} \label{ass:eigenvalues_K}
    The $d$ non-zero eigenvalues of $K$ are all distinct.
\end{assumption}
\begin{assumption} \label{ass:psi_mixing_y}
    The sequence $\{\byt\} \in \mathbb{R}^n$ is strictly stationary and $\psi$-mixing with coefficients such that $\sum_{\ell=1}^\infty \ell \, \psi^{1/2}(\ell) < \infty.$ Furthermore, $\Ex{\left| \byt \right|^4} < \infty$ elementwise.
\end{assumption}

Assumption \ref{ass:full_rank_xi} is equivalent to saying that $\Sigmaxi$ is full rank. If this assumption is violated, one can always remove the redundancy by redefining the factors. This assumption also implies that there exists $k \geq 1$ such that $\Sigmaxilagged{k}$ is full rank, which is relevant for Assumption \ref{ass:full_rank_K}. Assumption \ref{ass:stochastic_indepdence} is imposed for simplicity and is stronger than strictly necessary for deriving the asymptotics of optimal denoising, as long as certain conditions on the (lagged) covariance between signal and noise are satisfied. Assumption \ref{ass:full_rank_K} ensures that $\mathcal{M}$ is spanned by the eigenspace of $K$ associated with the non-zero eigenvalues, and Assumption \ref{ass:eigenvalues_K} is needed to derive a bound on estimation errors of the associated eigenvectors. As reported in \cite{Lam2011}, Assumption \ref{ass:psi_mixing_y} can be replaced by an $\alpha$-mixing condition, at the expense of additional technicalities.

\begin{theorem} \label{thm:asymptotics}
Given an observed sample $\left\{ \byt \right\}_{t=1}^T$ of length $T$ from a time series process satisfying Assumptions \ref{ass:full_rank_xi}-\ref{ass:psi_mixing_y}, the denoising error of optimal denoising $\bxthatopt = \Popthat \byt$ satisfies, for each fixed $t$, as $T \to \infty$,
\begin{equation}
    \begin{matrix*}[l]
        \norm{\bxthatopt - \bxt} = \bigOp{T^{-1/2}}  & \quad \mathrm{if} \quad \mathcal{M} \cap \mathcal{M}_\varepsilon = \emptyset, \\
        & \nonumber \\
        \norm{\bxthatopt - \bxt - \bepstc + \mathrm{Cov} \left[ \bepstc, \bepstperp \right] W_\perp ( \Sigmachi )^{-1} W_\perp^\top \bepstperp} = \bigOp{T^{-1/2}}  & \quad \mathrm{if} \quad \mathcal{M} \cap \mathcal{M}_\varepsilon \neq \emptyset.
    \end{matrix*}
\end{equation}
\end{theorem}

The proof of Theorem \ref{thm:asymptotics} is deferred to Appendix \ref{app:asymptotics}.
For the estimation of $\mathcal{M}$ we adapted the analysis of \citet{Bathia2010}, based on V-statistics, to the setting of finite-dimensional time series. Building on this result, we then derived the asymptotic behavior of optimal denoising by means of existing results in time series asymptotics and perturbation theory of eigenvalues and eigenvectors.

\section{Simulation studies} \label{sec:simulation}

This chapter uses simulations to assess three aspects of the denoising performance of (MSE-)optimal denoising: recovery of the unobserved signal, i.e.,~the actual denoising (Section~\ref{sec:results}), improved parameter estimation (Section~\ref{sec:parameter_estimation}) and improved one-step-ahead forecasting (Section~\ref{sec:forecasting}). The data-generating process of the simulated data is inspired by \citet{Bathia2010, Chen2022}. The performance of (MSE-)optimal denoising is primarily compared with orthogonal denoising, as described in Section~\ref{sec:denoising}.

However, note that both methods use a consistent estimator for the dynamic space (i.e.~the eigenspace of $K$ defined in Equation~\eqref{eq:K}). Alternatively, both methods could also be applied with an asymptotically biased estimate of $\mathcal{M},$ using the naive estimator of the covariance matrix $\Sigmay = \Sigmax + \Sigmaeps$ of the observed time series. The dominant noise directions are then erroneously attributed to the dynamic space. As a consequence, there is no orthogonal noise component that can be denoised or used to remove the parallel noise component. This means that neither orthogonal denoising nor (MSE-)optimal denoising using an asymptotically biased estimate of $\mathcal{M}$ are able to meaningfully remove noise from the observed time series, even asymptotically. Apart from Figure~\ref{fig:dependence_d_p}, we leave out the denoising results of the biased methods, with the understanding that in Figures~\ref{fig:dependence_n}-\ref{fig:dependence_thetas} they would have led to zero denoising on average. For parameter estimation and forecasting, we do include a biased method (see Figures~\ref{fig:parameter_estimation} and \ref{fig:forecasting_performance}).

\subsection{Setup} \label{sec:setup}
We consider a simulated time series $\{ \byt \}_{t=1}^T$ of length $T,$ where $\byt \in \mathbb{R}^n,$ and where $\byt = \bxt + \bepst + \bepstnull.$ The terms $\bxt$ and $\bepst$ correspond to the signal and noise of the time series considered throughout Section \ref{sec:methodology}. We have added a new term $\bepstnull,$ which will be specified later and which represents unstructured noise.

For the data-generating process of the signal we choose a factor-model form $\bxt = U \bxit,$ where the factors $\bxit$ follow a $d$-dimensional VAR($p$) process
\begin{equation} \label{eq:VAR_process}
\bxit = \sum_{s=1}^p A_s \bxi_{t-s} + \bet, \qquad \text{where} \quad \bet \thicksim \mathcal{N} \! \left( {\bm 0}, \Sigma_{\mathrm{\bf e}} \right).
\end{equation}
To ensure generality of the simulation results, we use an algorithm that randomly generates VAR($p$) processes that are stationary, have a diagonal covariance matrix $\Sigmaxi,$ and for which $\Tr{\Sigmaxi}=1.$ Furthermore, we demand that $\Tr{\Sigma_{\mathrm{\bf e}}}$ is  within a specified range, making it easier to compare asymptotic behavior of different VAR($p$) processes. The details of this algorithm are given in Appendix \ref{app:simulation_details}. Also note that this form of the VAR($p$) process automatically ensures that $\Ex{\bxt} = U \Ex{\bxit}={\bm 0}.$

The factor loading matrix $U$ has elements
\[
U_{ai} = \sqrt{1 - \lambda} \, \sqrt{\frac{2}{n}} \cos \left( \frac{2 \pi i}{n}a \right), 
\]
where $a=1,\ldots,n$ and $i=1,\ldots,d.$ The parameter $\lambda \in (0,1)$ normalizes the signal such that $\tracevar{\bxt} = 1 - \lambda.$ The choice of $U$ means that the signal $\bxt$ is a linear combination of discretized trigonometric functions, represented at an increasingly fine resolution as $n$ increases. We emphasize that the use of basis functions of $\mathcal{M}$ that are smooth is purely a matter of convenience and does not affect the denoising. In fact, the denoising procedure is equivariant under a permutation of the coordinate index $a$: randomly shuffling the observed data $\byt$ in the index $a$ does not change the denoising procedure in any way and leads to an identical denoised signal. Hence, a choice for basis functions that are smooth does not impact the generality of the simulation results. This is important, because real-world high-dimensional time series data often lacks this smoothness property (e.g., stock market prices).

For the structured part of the noise we also use a factor-model form $\bepst = W \bxiepst,$ but assume that the factors are serially independent: $\bxiepst \thicksim \mathcal{N} \! \left( {\bm 0}, \Sigmaxieps \right)$ are i.i.d.~over time, with
\[
\Sigmaxieps = \mathrm{diag} \left( \sigma_{\bvarepsilon, 1}^2, \ldots, \sigma_{\bvarepsilon, d_\varepsilon}^2 \right) \quad \text{and} \quad \sigma_{\bvarepsilon, j} = \frac{c_\varepsilon^{j-1}}{\sqrt{\sum_{k=1}^{d_\varepsilon} c_\varepsilon^{2(k-1)}}}
\]
with $c_\varepsilon \in (0,1)$.
This noise consists of $d_\varepsilon$ modes whose contributions to the total variance decrease geometrically, while the total variance is normalized to $\Tr{\Sigmaxieps}=1.$ The factor loadings are given by
\[
W_{aj} = \sqrt{\lambda} \sqrt{1-\delta} \, \sqrt{\frac{2}{n}} \left\{ \cos (\theta_j) \cos \left( \frac{2 \pi j}{n}a \right) + \sin (\theta_j) \sin \left( \frac{2 \pi j}{n}a \right) \right\},
\]
where $a=1,\ldots,n$ and $j=1,\ldots, d_\varepsilon.$ The normalization factors are chosen such that $\tracevar{\bepst} = \lambda (1 - \delta).$ The angles $\theta_j$ specify the orientation of the noise space $\mathcal{M}_\varepsilon$ relative to the dynamic space $\mathcal{M}.$ If $\theta_j=0$ (for $j \leq d$), the $j$-th noise mode lies inside the dynamic space. If $\theta_j = \pi/2,$ the noise mode is in the orthogonal complement of $\mathcal{M}.$

Finally, for the unstructured part of the noise all elements are independent across both coordinates and time. We take $\bepstnull \thicksim \mathcal{N} \! \left( {\bm 0}, \frac{\lambda \delta}{n} I_n \right),$ independently over time, making $\tracevar{\bepstnull}=\lambda\delta.$

With the above specification of the data-generating process, the total variance of the observed time series is conveniently normalized to one, $\tracevar{\byt} = 1.$ Most of our simulations do not include unstructured noise ($\delta=0$), implying $\tracevar{\bepst} = \lambda,$ which leads to the interpretation of the parameter $\lambda$ as the noise fraction of the total variance of the observed time series. In a similar fashion, $\delta = \tracevar{\bepstnull} / \tracevar{\bepst}$ is the fraction of unstructured noise variance in the total noise variance.

\begin{table}[]
    \centering
    \begin{tabular}{|c|c|}
        \hline
         {\bf Parameter} & {\bf Default value} \\ \hline
         $d$ & 4 \\
         $p$ & 2 \\ 
         $n$ & 50 \\ \
         $d_\varepsilon$ & 8 \\
         $c_\varepsilon$ & 2/3 \\
         $\lambda$ & 0.2 \\
         $\delta$ & 0.0 \\ 
         $T$ & 400 \\
        \hline
    \end{tabular}
    \caption{Default values of simulation parameters. Unless stated otherwise, these values are used in the simulations presented in Section \ref{sec:simulation}.}
    \label{tab:default_values_simulation}
\end{table}

To assess the general applicability of optimal denoising in our simulations, we vary one parameter at a time in the data-generating process while keeping the other parameters fixed. Table \ref{tab:default_values_simulation} shows the default values of the simulation parameters. Unless stated otherwise, these values are being used in the simulations presented here. When generating the factors $\bxit$ of the VAR($p$) process underlying the signal, we initialize the factors at zero and we use a burn-in phase of 200 time steps.

\subsection{Denoising results} \label{sec:results}


\begin{figure}
    \centering
    \includegraphics[scale=0.8]{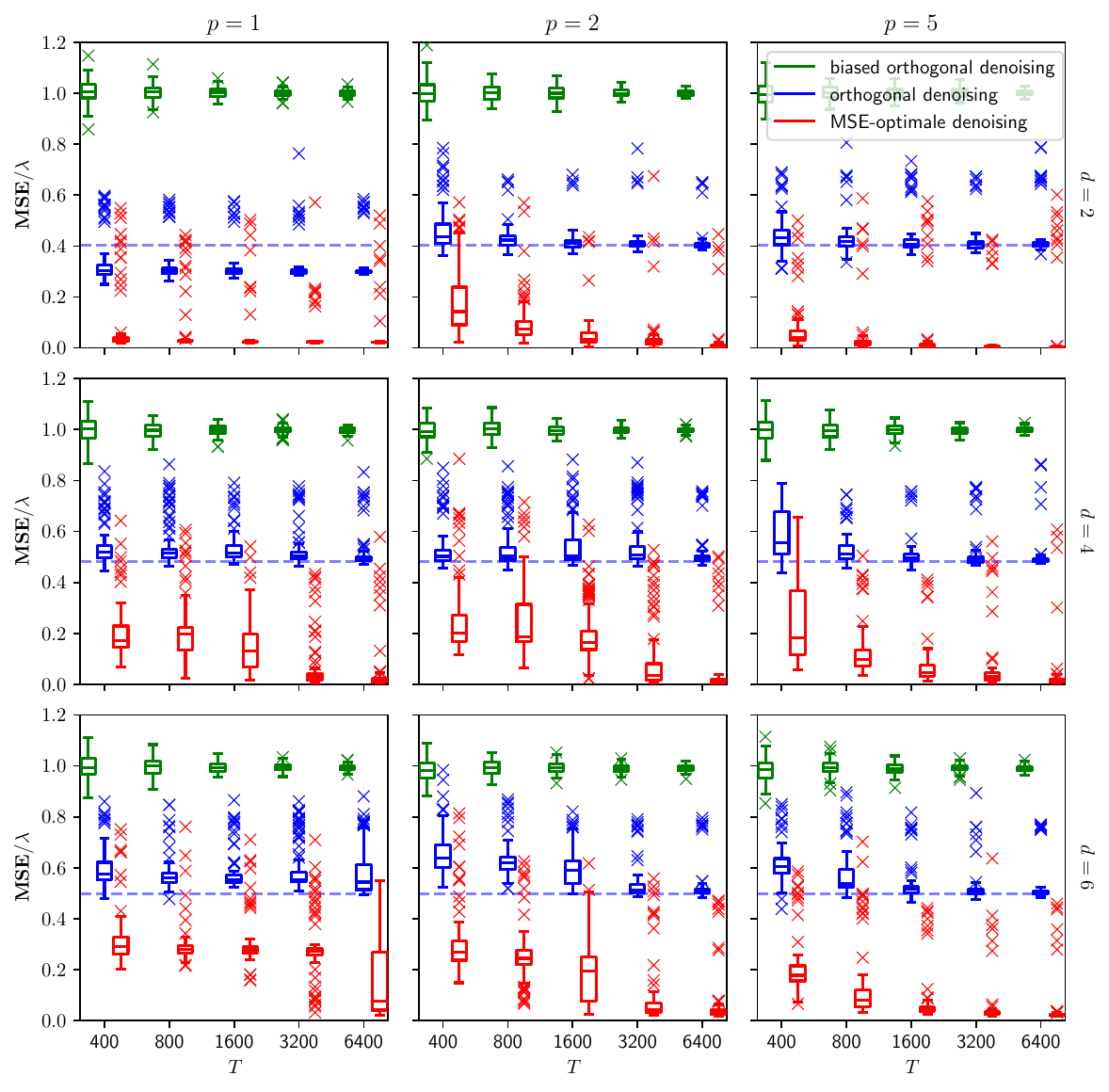}
    \caption{(asymptotic) denoising performance, measured as $\mse/\lambda,$ as a function of the dimension $d$ and the lag order $p$ of the VAR($p$) process. Each boxplot consists of 100 independent simulations. For the sake of a baseline, we have included (in green) orthogonal denoising based on the standard asymptotically biased estimate of the covariance matrix of the signal.}
    \label{fig:dependence_d_p}
\end{figure}

Given a time series $\{ \byt \}_{t=1}^T$, estimation and denoising proceeds as described in Section \ref{sec:estimation}. Since the specific choice of $K$ in Equation \eqref{eq:K} has little influence on the estimation of $\mathcal{M}$, we fix $K$ by setting $c_1=c_2=1$ and all other coefficients to zero \citep{Bathia2010, Lam2011}. For the bootstrap procedure used to estimate $d,$ described in Appendix \ref{app:bootstrap}, we set $B=300$ and $\alpha=0.05.$ To estimate $d_\varepsilon$ we set the cutoff to $\tau_\perp=0.95.$

For both orthogonal denoising and optimal denoising, we measure denoising performance in terms of the estimated mean squared error
\[
\mathrm{\bf MSE} = \frac{1}{T} \sum_{t=1}^T \normsq{\bxthatden - \bxt}.
\]
In case of trivial denoising, $\Pden=I_n$, so the observed time series is unchanged and $\Ex{\mathrm{\bf MSE}} = \lambda.$ Thus, the no-denoising case yields an $\mathrm{\bf MSE}$ of $\lambda$, which we use as a baseline. Values with $\mathrm{\bf MSE}/\lambda < 1$ indicate beneficial denoising, whereas values with $\mathrm{\bf MSE}/\lambda > 1$ indicate that denoising is counterproductive.

In Figure~\ref{fig:dependence_d_p} we report the denoising performance, measured by the normalized mean squared error $\mse / \lambda,$ as a function of the time series length $T,$ for different values of the dimension $d$ and the lag order $p$ of the VAR($p$) process. Each boxplot is based on 100 independent simulations. MSE-optimal denoising outperforms orthogonal denoising for all VAR($p$) specifications. In accordance with our theoretical analysis, the MSE-optimal denoiser converges to zero denoising error, whereas orthogonal denoising exhibits an irreducible denoising error equal to the total variance of $\bepstpara.$ Both methods produce a large number of outliers, which can be attributed to misestimation of the VAR dimension $d$ by the bootstrap procedure (see Appendix~\ref{app:bootstrap} for details).

As a baseline, Figure~\ref{fig:dependence_d_p} also includes results for orthogonal denoising based on the naive, asymptotically biased, estimate of the covariance matrix of the signal. On average, this method does not remove any noise from the observed time series. As explained in the introduction of Section~\ref{sec:simulation}, MSE-optimal denoising would give a similar poor performance when based on the (inconsistent) naive estimate of $\mathcal{M}$.

Figures~\ref{fig:dependence_n}-\ref{fig:dependence_thetas} in Appendix~\ref{app:additional_simulation_results} report results from the same experimental design with other hyperparameters varied. In Figure~\ref{fig:dependence_n} we find that the dimension $n$ of the observed time series $\{\byt\}$ affects denoising performance only when $n$ ($\lesssim 10$). This reflects that the simulated dynamic and noise spaces have dimensions  $d=4$ and $d_\varepsilon=8$, respectively. Identifying all $d + d_\varepsilon = 12$ modes therefore requires $n=12$ observed dimensions. For smaller $n$, complete noise removal is impossible, even asymptotically. 

Figure~\ref{fig:dependence_trace_sigma_e} shows how denoising performance varies with $\text{Tr}\left[ \Sigma_{\mathrm{\bf e}} \right],$ the noise-to-signal ratio under the normalization $\tracevar{\bxit}=1.$ Larger noise-to-signal ratios require longer time series for denoising performance to get close to its asymptotic value. When the noise-to-signal ratio is large and the time series is short, denoising can be harmful ($\mse /\lambda > 1$): the denoised signal $\bxtden$ then is a less accurate approximation of $\bxt$ than the observed time series $\{\byt\}$. This is due to overfitting and may possibly be remedied with a form of (Ridge) regularization. 
We attribute this behavior to overfitting, which may be mitigated by  regularization. This extension is beyond the scope of the present paper, as regularization is unnecessary under the assumed finite dimensionality of the signal and noise subspaces.

Figure~\ref{fig:dependence_lambda} examines the effect of $\lambda,$ the noise fraction in the observed time series,  for time series of the default length $T=400.$ For very small noise fractions ($\lambda = 0.01$) orthogonal denoising outperforms MSE-optimal denoising. However,  the right panel of Figure~\ref{fig:dependence_lambda} shows that  that this difference is negligible in absolute terms. Even for very large noise fractions ($\lambda = 0.99$) optimal denoising removes about 70\% of the noise on average.

Figure~\ref{fig:dependence_delta} investigates the effect of changing the ratio of structured noise $\bepst$ and unstructured noise $\bepstnull$ by varying $\delta.$ 

In Figure~\ref{fig:dependence_thetas} we vary the orientation of the noise space relative to the dynamic space. In the left panel the second noise component is parallel to the dynamic space. In accordance with Theorem~\ref{thm:remaining_noise}, MSE-optimal denoising then exhibits an asymptotically irreducible denoising error. In the right panel the entire noise space lies in the orthogonal complement of the dynamic space and MSE-optimal denoising performs similarly to orthogonal denoising.

\subsection{Improved parameter estimation} \label{sec:parameter_estimation}
In this section we use the simulation setup to investigate whether model-agnostic optimal denoising can be used as a preprocessing step that improves parameter estimation of the underlying low-dimensional model that drives the serial dependence, which for our simulations is the VAR($p$) process defined in Equation \eqref{eq:VAR_process}. It should be noted that this representation of the VAR($p$) process cannot be estimated without bias, because the dynamic space $\mathcal{M}$ is spanned by the columns of $V$ rather than the columns of $U.$ If we define factors $\bzetat = V^\top \bxt,$ the VAR($p$) process of Equation \eqref{eq:VAR_process} can be written in an alternative representation
\begin{equation} \label{eq:VAR_process_V_represenation}
\bzetat = \sum_{s=1}^p B_s \bzeta_{t-s} + \bet', \qquad \text{where} \quad \bet' \thicksim \mathcal{N} \! \left( {\bm 0}, \Sigma_{\mathrm{\bf e'}} \right),
\end{equation}
and where $B_s = V^\top U A_s U^\top V$ and $\bet' = V^\top U \bet$ link both representations. Given a chosen denoising method represented by $\Pden,$ we use the denoised factors $\bzetat^{\rm den} = V^\top \bxtden$ to estimate the parameters of the VAR($p$) process by means of multivariate least squares estimation \citep{Lutkepohl2005}. We fit the models with the \verb+statsmodels+ Python module and in order to compare the parameter values of $B_s$ for $s=1,2,\ldots,p$ the lag order $p$ is treated as known and set to its true oracle value.

\begin{figure}
    \centering
    \includegraphics[scale=0.8]{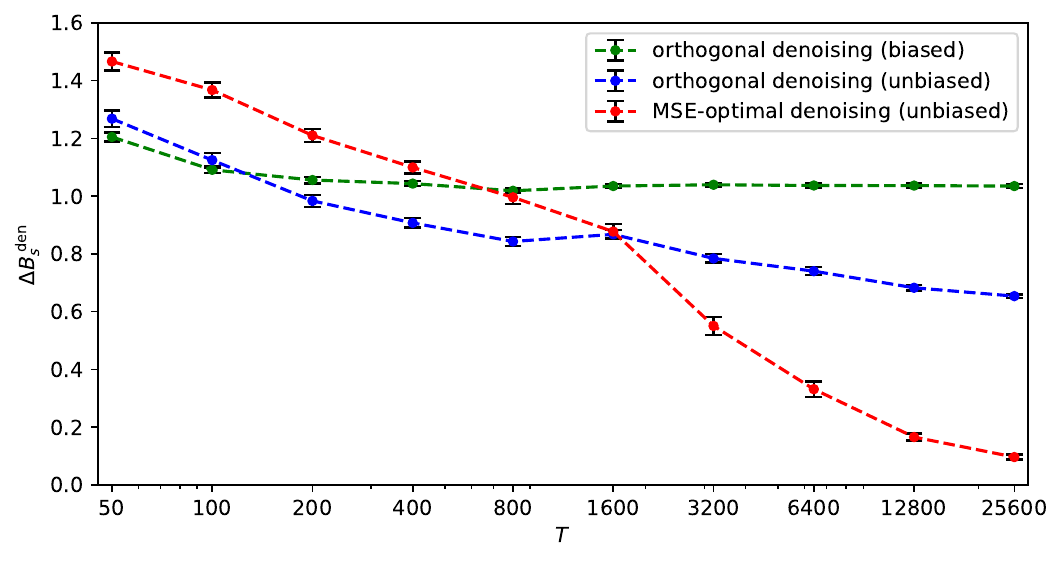}
    \caption{the (asymptotic) behavior of the parameter estimation error after denoising. Each point is the average of 100 independent simulations and the standard error. We have assumed the lag order is known, since the plotted measure $\Delta B_s^{\rm den}$ only makes sense if the lag order of the estimated model equals the true $p.$}
    \label{fig:parameter_estimation}
\end{figure}

As a measure of parameter estimation performance we define
\[
\Delta B_s^{\rm den} = \Ex{\frac{1}{p} \sum_{s=1}^p \frac{|| \hat{B}_s^{\rm den} - B_s ||_F}{|| B_s ||_F}},
\]
where $\hat{B}_s^{\rm den}$ denotes the estimator of the VAR($p$) coefficient matrix at las $s$ after denoising. Figure~\ref{fig:parameter_estimation} compares the performance of optimal denoising, orthogonal denoising and no denoising ($\Pden = I_n$). We find that parameter estimation becomes consistent after optimal denoising, with the estimation error converging to zero as the time series length increases. In contrast, orthogonal denoising and no denoising do not yield consistent parameter estimates.

\subsection{Improved one-step-ahead forecasting} \label{sec:forecasting}
Subsequently, the estimated VAR($p$) model can be used to forecast the observed time series. We consider one-step-ahead forecasting,
\[
\byt^{(1)} = V \sum_{s=1}^p B_s V^\top \Pden {\bf y}_{t-s},
\]
and measure forecasting performance in terms of the mean squared error,
\[
\Ex{|| \byt - \hat{\bf y}_t^{(1)}||^2},
\]
where $\hat{\bf y}_t^{(1)}$ is the sample version of the one-step-ahead forecast. Here, we do not assume the lag order $p$ to be known. Instead, we estimate it using order selection with the Akaike information criterion and \verb+statsmodels+' default maximum number of lags to be checked ($12 * (T/100)^{1/4}$).

In Figure~\ref{fig:forecasting_performance} we see that optimal denoising is again consistent, as its forecast mean squared error  converges to the theoretical lower bound $\lambda + (1-\lambda) \Tr{\Sigma_\mathrm{\bf e}}.$ Moreover, it outperforms orthogonal denoising and forecasting without denoising. We also note that the diverging forecast errors for short time series can be attributed to  poor lag order estimation. Finally, we note that orthogonally denoised and non-denoised forecasting both improve for large $T$ when the lag order is not given and must be estimated. The lag order is typically overestimated, which gives the VAR model, estimated with an asymptotic bias, more degrees of freedom to capture the serial dependence in the data.

\begin{figure}
    \centering
    \includegraphics[scale=0.7]{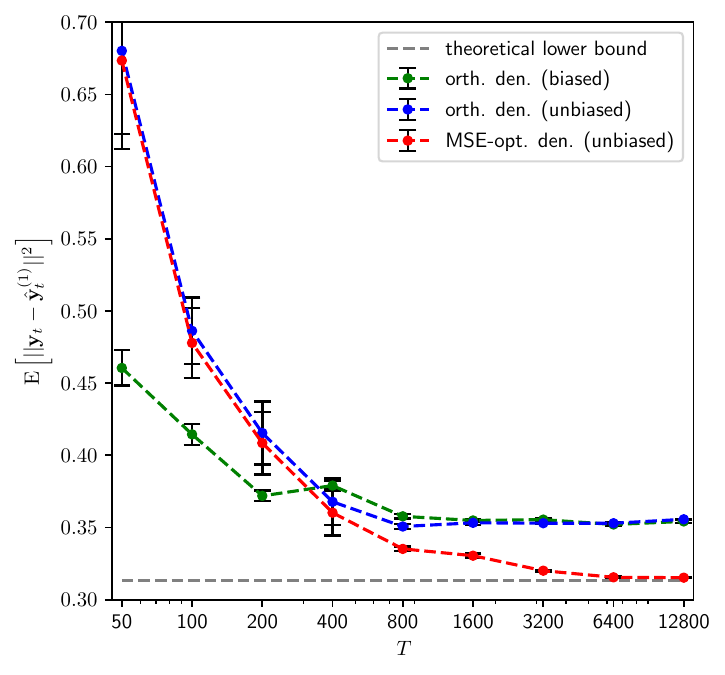}
    \caption{the (asymptotic) behavior of the mean squared error of one-step-ahead forecasting after denoising. Each point is the average of 100 independent simulations and the standard error.}
    \label{fig:forecasting_performance}
\end{figure}

\section{Applications} \label{sec:applications}

We now illustrate the proposed noise reduction method using two real-world high-dimensional time series datasets: a panel of daily stock returns and a panel of monthly macroeconomic indicators. Both settings are commonly modeled through a low-dimensional factor structure, but they differ substantially in their signal-to-noise characteristics, which allows us to assess the method under conditions that are quite different from the simulation study of Section~\ref{sec:simulation}. Throughout, we use the same estimation procedure as in Section~\ref{sec:estimation}, with $K$ defined via $c_1=c_2=1$ (i.e.~$k_0=2$) and, when the bootstrap procedure of Appendix~\ref{app:bootstrap} is used, $B=300$ bootstrap samples and significance level $\alpha=0.05.$ The dimension of the orthogonal noise component $d_\perp$ is estimated with the explained-variance rule of Section~\ref{sec:estimation}, using cutoff $\tau_\perp=0.95.$

\subsection{Data} \label{sec:application_data}
The first dataset consists of daily log returns of the $n=459$ constituents of the S\&P~500 index for which a complete price history is available over the sample period from 5 January 2015 through 13 February 2026, giving a time series of length $T=2795.$ This is a typical example of a large cross-section of financial return series in which any common dynamic structure, if present at all, is known to be weak and easily obscured by idiosyncratic and market-wide noise.

The second dataset is a panel of the first $n=20$ monthly US~macroeconomic indicators (measures of output, income and industrial production) in the Federal Reserve Economic Data (FRED) 2026-02-MD snapshot obtained from the macroeconomic database maintained by the Federal Reserve Bank of St.\ Louis, containing monthly data from January 1959 until January 2026.  Although the dataset contains 127 different indicators, we reduced the analysis to the leading indicators listed in the first 20 columns, ignoring the secondary and sometimes incomplete indicators in the later columns. 
For each series we compute first differences and standardize the result to unit sample variance, giving a full sample of length $T=802,$ spanning February 1959 to November 2025, some of the 20 indicators being available only up until then in the dataset. For the leave-one-year-out cross-validation exercise of Section~\ref{sec:application_forecasting} we restrict attention to the 65 complete calendar years contained in this sample, from 1960 up to and including 2024 ($T=780$ months). Unlike the daily returns data of the S\&P~500 index, monthly macroeconomic aggregates are widely believed to be governed by a small number of common dynamic factors that account for a substantial share of their total variance.

For both datasets we inspect the eigenvalues of $\Khat$ (Equation~\eqref{eq:K}) and apply the bootstrap procedure of Appendix~\ref{app:bootstrap} to estimate the dimension $d$ of the dynamic space. In both cases the eigenvalues of $\Khat$ show a pronounced elbow, with only the first few eigenvalues clearly separated from a long tail of eigenvalues close to zero, yet the bootstrap counter of Algorithm~\ref{alg:bootstrap} behaves non-monotonically in $d_0$ and does not settle on a stable estimate of $d,$ a pattern consistent with the underestimation of $d$ for noisy or moderately sized samples documented in Appendix~\ref{app:bootstrap}. We therefore follow the alternative suggested by \citet{Bathia2010} themselves and set $d$ by visual inspection of the eigenvalues of $\Khat,$ rather than by the sequential bootstrap algorithm. For both applications this leads us to choose $d=4,$ which, together with $p=2$ for the VAR($p$) model used in forecasting, coincides with the default values used throughout the simulation study (Table~\ref{tab:default_values_simulation}), making the two sets of results easier to compare.

\subsection{Denoising results} \label{sec:application_denoising}
\begin{table}[]
    \centering
    \begin{tabular}{|l|c|c|}
        \hline
        {\bf Method} & {\bf S\&P~500 log-returns} & {\bf Macroeconomic panel} \\ \hline
        biased orthogonal denoising & 45.82\% & 78.07\% \\
        orthogonal denoising (unbiased) & 39.96\% & 73.72\% \\
        MSE-optimal denoising (unbiased) & 10.22\% & 46.29\% \\
        \hline
    \end{tabular}
    \caption{fraction of the total variance of $\byt$ that is attributed to the estimated signal $\bxtden$ ($\Tr{\Var{\bxthatden}}/\Tr{\Sigmayhat}$) by each denoising method, for both applications, using the handpicked dimension $d=4.$}
    \label{tab:application_denoising}
\end{table}
Table~\ref{tab:application_denoising} reports, for both datasets, the fraction of the total variance of $\byt$ that each denoising method attributes to the signal $\bxtden.$ For both applications this fraction drops sharply when moving from orthogonal to MSE-optimal denoising: from about 40\% to about 10\% for the S\&P~500 returns, and from about 74\% to about 46\% for the macroeconomic panel. This is consistent with our interpretation of the difference between the two methods: orthogonal denoising treats the entire component of $\byt$ within $\mathcal{M}$ as signal, even though part of it may be a parallel noise component $\bepstpara$ 
correlated with the orthogonal noise $\bepstperp.$ MSE-optimal denoising uses this correlation to remove part of $\bepstpara$ as well. That this reclassification is so pronounced in both datasets suggests that real-world noise is often not orthogonal to the dynamic space, reinforcing the relevance of the optimal denoising approach developed in this paper.

The effect is considerably stronger for the S\&P~500 returns than for the macroeconomic panel. Only about one tenth of the total variance of the daily returns is attributed to the signal after MSE-optimal denoising, compared with roughly half of the total variance for the macroeconomic panel. This is in line with the common view that daily stock returns are hard to predict, so that little of their variance can be attributed to a low-dimensional dynamic component, whereas monthly macroeconomic aggregates are widely believed to share a more substantial common dynamic component. Because the fraction of the variance of the S\&P~500 returns identified as signal by MSE-optimal denoising is so small, any downstream model fitted to the corresponding denoised factors is estimated on what is effectively very little signal. Indeed, one-step-ahead forecasts of the denoised S\&P~500 returns based on a VAR($p$) model, for any of the three denoising methods of Table~\ref{tab:application_denoising}, do not meaningfully improve on the total variance of $\byt$ itself, so that no method can reliably be said to outperform the others. We therefore do not report forecasting results for the S\&P~500 application and instead turn to one-step-ahead forecasting for the macroeconomic panel, where the signal fraction is large enough for the comparison to be informative.

\subsection{One-step-ahead forecasting} \label{sec:application_forecasting}
\begin{figure}
    \centering
    \includegraphics[scale=0.48]{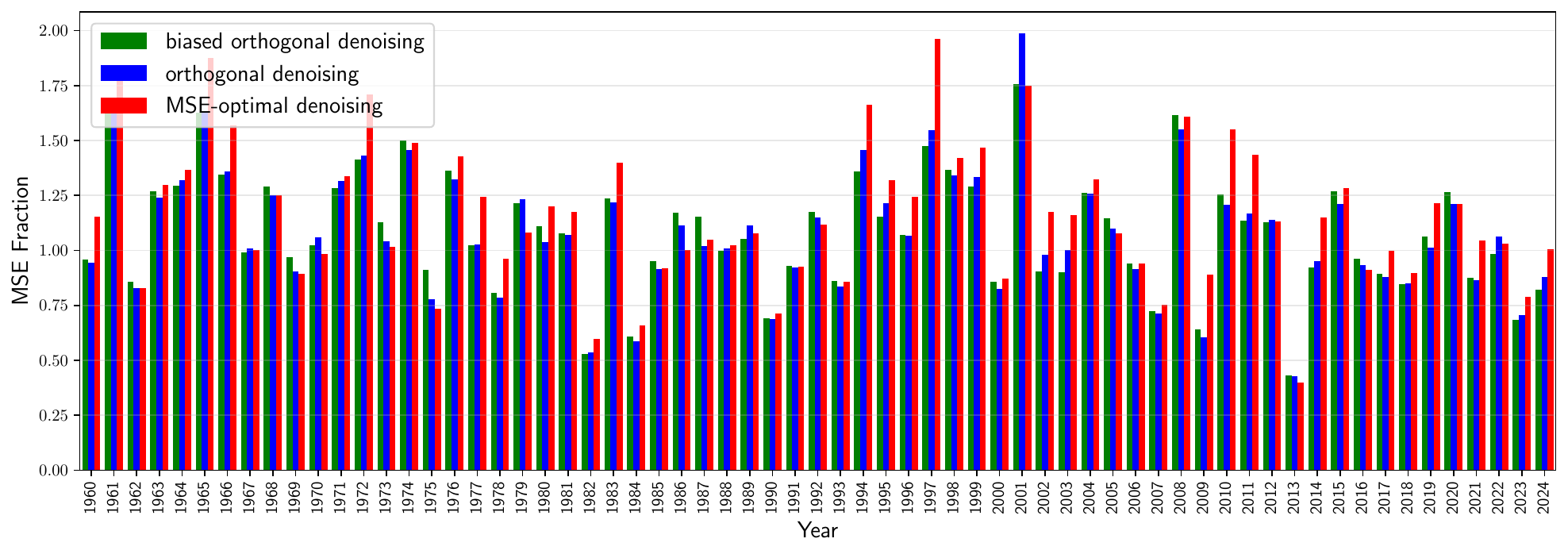}
    \caption{one-step-ahead forecasting MSE (as a fraction of the total variance that year) for the macroeconomic panel, per calendar year. Each year is left out of the training sample in turn and the corresponding VAR($p$) model, with $p=2$ and (handpicked) $d=4,$ is used to forecast that year.}
    \label{fig:application_macrodata_per_year}
\end{figure}
For the macroeconomic panel we assess one-step-ahead forecasting performance of the VAR($p$) model, with $p=2$ and $d=4,$ estimated on the denoised factors of each of the three methods in Table~\ref{tab:application_denoising}. We use a leave-one-year-out cross-validation scheme: for each of the 65 complete calendar years in the sample, that year is removed from the training data, the model is re-estimated on the remaining months, and the mean squared forecast error is computed for the left-out year. Figure~\ref{fig:application_macrodata_per_year} shows the resulting mean squared error for every calendar year and every method, while Table~\ref{tab:application_forecasting} reports the average performance over all 65 years, together with the corresponding standard error.

\begin{table}[]
    \centering
    \begin{tabular}{|l|c|}
        \hline
        {\bf Method} & {\bf Average MSE (std.~error)} \\ \hline
        biased orthogonal denoising & 22.90 (11.29) \\
        orthogonal denoising (unbiased) & 22.36 (10.83) \\
        MSE-optimal denoising (unbiased) & 23.47 (10.87) \\
        \hline
    \end{tabular}
    \caption{average one-step-ahead forecasting MSE for the macroeconomic panel over the 65 leave-one-year-out folds, with the standard error of the mean across folds in parentheses.}
    \label{tab:application_forecasting}
\end{table}

Unlike in the simulation study of Section~\ref{sec:forecasting}, the three methods perform comparably on average, with overlapping standard errors, and MSE-optimal denoising does not systematically improve on orthogonal denoising here. Figure~\ref{fig:application_macrodata_per_year} shows that this holds fairly uniformly across calendar years, in that the three methods track each other closely from year to year. This is a markedly different picture from the simulation results of Figure~\ref{fig:forecasting_performance}, where MSE-optimal denoising clearly dominates. A likely explanation is that, unlike the simulated data, the macroeconomic panel is not generated by a low-dimensional VAR($p$) process with additive white noise. As a result any of the three (denoised) representations of the data are likely to be similarly misspecified as inputs to a linear VAR forecasting model, leaving comparatively little room for noise reduction to improve out-of-sample forecasts.

\section*{Acknowledgments}

The authors are grateful to Peter Boswijk, Sven Otto, Peter Zadrozny and participants at various seminars and conferences, including at the  10th International Workshop on Applied Probability in Thessaloniki, the Bernoulli-IMS 11th World Congress in Probability and Statistics in Bochum, for their comments and suggestions.

\newpage
\appendix

\section{Bootstrap tests for estimating $d$} \label{app:bootstrap}
We use the bootstrap test proposed by \citet{Bathia2010} to estimate the dimension of the dynamic space $d.$ In this appendix we describe the method and evaluate its performance on the simulated data of Section~\ref{sec:simulation} and assess its effect on our denoising approach.

Using the notation of \citet{Bathia2010}, we denote the non-zero eigenvalues of the matrix $K,$ which was defined in Equation~\eqref{eq:K}, by $\theta_1 \geq \theta_2 \geq \ldots \geq \theta_d > 0.$ When we test the null hypothesis $H_0: d = d_0,$ we focus on testing its implication $\theta_{d_0 + 1} = 0$ under the null hypothesis. Let $\hat{\theta}_1 \geq \hat{\theta}_2 \geq \ldots \geq 0$ be the eigenvalues of $\Khat,$ the sample version of $K$ based on an observed sample $\{ \byt \}_{t=1}^T,$ which typically has $n > d$ strictly positive eigenvalues. We reject the null hypothesis if we observe $\hat{\theta}_{d_0 + 1} > l_\alpha,$ where $l_\alpha$ is a threshold value to be determined by a bootstrap test and $\alpha \in (0,1)$ is the significance level. \citet{Bathia2010} suggest visual inspection of the estimated eigenvalues to choose $d_0.$ Instead, we perform a series of repeated hypothesis tests for increasing value of $d_0,$ starting at $d_0=1,$ until the null hypothesis is not rejected. Algorithm~\ref{alg:bootstrap} gives a detailed description of the procedure we use.

\begin{algorithm}
\caption{Bootstrap method to estimate $d$}\label{alg:bootstrap}
\begin{algorithmic}[1]
\Require observed sample $\{ \byt \}_{t=1}^T;$ significance level $\alpha \in (0,1);$ bootstrap sample size $B.$
\State $d_0 \gets 0; \texttt{reject} \gets \textbf{True}$
\While{\texttt{reject}}
    \State $d_0 \gets d_0 + 1$
    \State Let $\widehat{V}_{d_0}$ be the $(n \times d_0)$-matrix whose columns are the first $d_0$ eigenvectors of $\Khat.$
    \State $\hat{\bf y}_t \gets \bar{\bf y}_t = \widehat{V}_{d_0} + \hat{\eta}_t,$ where $\hat{\eta}_t = \widehat{V}_{d_0}^\top (\byt - \bar{\bf y}_t)$ and $\bar{\bf y}_t = \tfrac{1}{T} \sum_{t=1}^T \byt.$
    \State $\hat{\bm \varepsilon}_t \gets \byt - \hat{\bf y}_t$
    \State $\texttt{counter} \gets 0$
    
    \For{$j = 1,2,\ldots,B$}

        \State $\byt^* \gets \hat{\bf y}_t + \bepst^*,$ with $\bepst^*$ randomly drawn (with replacement) from $\{\hat{\bm \varepsilon}_1, \hat{\bm \varepsilon}_2, \ldots, \hat{\bm \varepsilon}_T \}.$
        \State Compute $K^*$ based on sample $\{ \byt^* \}_{t=1}^T$ and let $\theta^*_{d_0 +1}$ be its $(d_0 +1)$th eigenvalue.
        \If{$\hat{\theta}_{d_0 + 1} > \theta^*_{d_0 +1},$}
            \State $\texttt{counter} \gets \texttt{counter} + 1$
        \EndIf
    \EndFor

    \If{$c < \lfloor \alpha B \rfloor$}
        \State $\texttt{reject} \gets \textbf{False}$
    \EndIf
\EndWhile
\State \Return $d_0$
\end{algorithmic}
\end{algorithm}

In Figure~\ref{fig:bootstrap_analysis_ds} we analyze the performance of the bootstrap method on the simulated data of Section~\ref{sec:simulation} with the default values described in Table~\ref{tab:default_values_simulation} and in particular $d=4.$ We use $B=300$ bootstrap samples and vary the significance level $\alpha.$ We observe that for long time series the fraction of incorrect estimates is close to the nominal significance level $\alpha$. In this regime the dimension of the dynamic space $\mathcal{M}$ is typically overestimated. For shorter time series (see $T=100$) the fraction of misestimated dimensions is larger, the dimension typically being underestimated.

\begin{figure}
    \centering
    \includegraphics[scale=0.6]{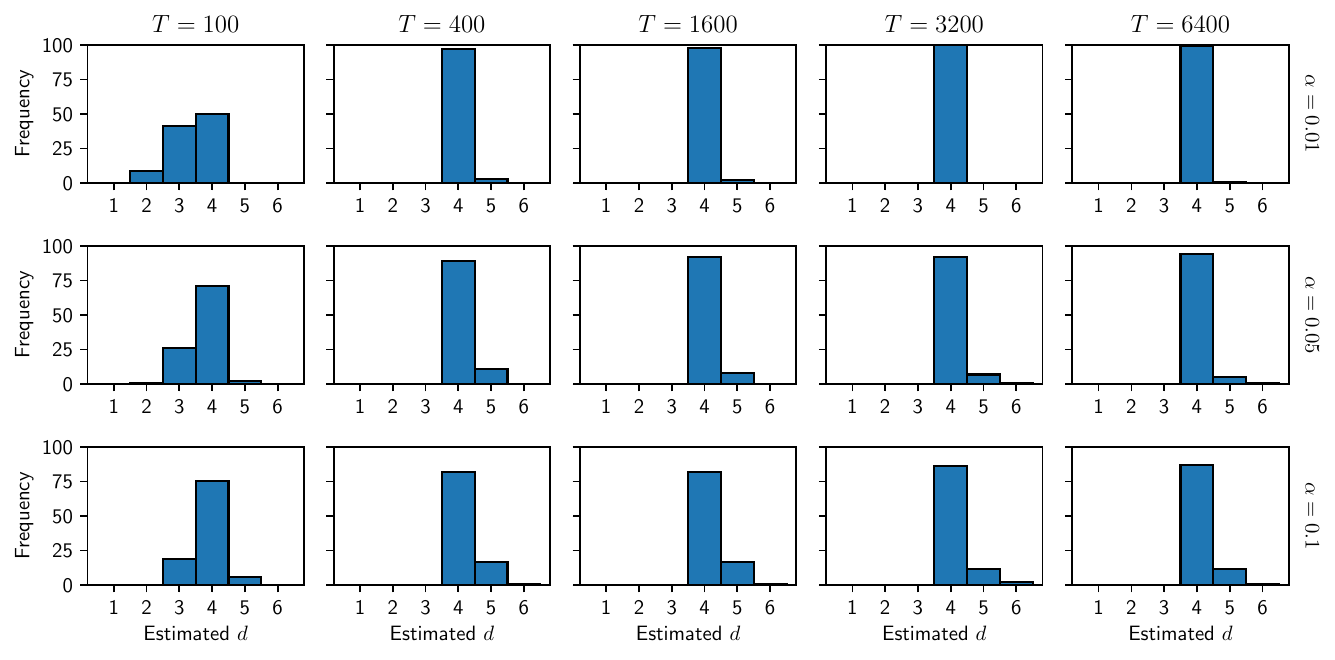}
    \caption{histogram (based on 100 independent simulations) of estimated dimension of the dynamic space using the bootstrap Algorithm~\ref{alg:bootstrap}, for varying significance level $\alpha$ and time series length $T.$ The true dimension in the simulated data is $d=4.$}
    \label{fig:bootstrap_analysis_ds}
\end{figure}

Figure~\ref{fig:bootstrap_analysis_combined} studies the effect of misestimation of $d$ on orthogonal and MSE-optimal denoising. In the left panel of Figure~\ref{fig:bootstrap_analysis_combined} we see the typical boxplots of the normalized MSE after denoising, based on 100 independently simulated time series, analogous to the panel with $p=2$ and $d=4$ in Figure~\ref{fig:dependence_d_p}. In the right panel we report the estimated value of $d$ for each individual time series. Recall that we use $B=300$ and $\alpha=0.05$ for the bootstrap procedure. We observe that almost all outliers are due to incorrect estimates of the dimension $d.$ We also observe that the effect of this on orthogonal and MSE-optimal denoising is comparable. Finally, we see that for short time series ($T=100$) the performance drop due to faulty estimation of $d$ is less pronounced. This suggests using a small significance level, e.g., $\alpha=0.01,$ which reduces the fraction of erroneous estimates for long time series. 

\begin{figure}
    \centering
    \includegraphics[scale=0.58]{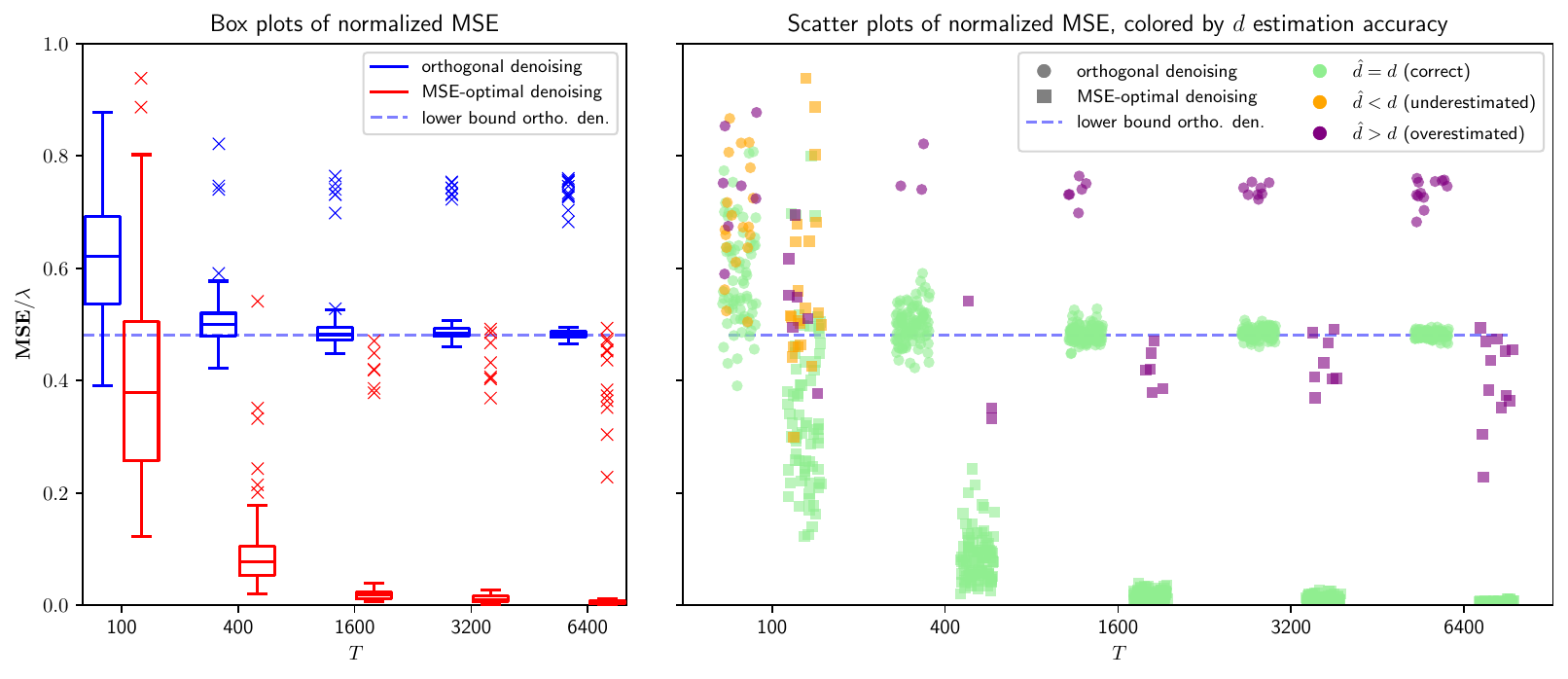}
    \caption{denoising performance, measured as $\mse/\lambda,$ for a VAR($p$) process with $d=4$ and $p=2.$ Each boxplot (left panel) consists of 100 independent simulations. In the right panel we specify the estimation of $d$ for each individual time series.}
    \label{fig:bootstrap_analysis_combined}
\end{figure}

\section{Asymptotic analysis} \label{app:asymptotics}

The purpose of this section is to establish the asymptotic behavior of optimal denoising, as specified in Theorem \ref{thm:asymptotics}. We start by reviewing a result on V-statistics, which underpins the asymptotic analysis of the estimator of $K.$ This analysis builds on on \citet{Bathia2010} and is adapted here to finite-dimensional time series.

\subsection{V-statistics and their asymptotic behavior}
Suppose $\bZt \in \mathcal{V}$ is a sequence of stationary random variables with marginal distribution function $P_Z(z)$ and $\mathcal{V}$ a finite-dimensional vector space. Given a function $\bphi: \mathcal{V}^m \to \mathcal{V}$ that is symmetric in each of its $m \geq 2$ arguments, we define the quantity
\begin{equation} \label{eq:theta_V_statistic}
\btheta(P_Z) = \int_{\mathcal{V}^m} \bphi(\bz_1, \ldots, \bz_m) \prod_{j=1}^m P_Z(\dd \bz_j),
\end{equation}
assuming that $P_Z(\cdot)$ is such that $\norm{\btheta(P_Z)} < \infty.$ Note that we integrate over the product measure, treating the arguments of $\bphi$ as independent copies of $\bZt$. Given a sequence $\bZ_1, \ldots, \bZ_T$ of $T$ observations, the V-statistic
\[
V_T = \frac{1}{T^m} \sum_{i_1=1}^T \cdots \sum_{i_m=1}^T \bphi(\bZ_{i_1}, \ldots, \bZ_{i_m})
\]
is a consistent estimator of $\btheta(P_Z)$.
The Hoeffding decomposition \citep{Vaart_1998} of $V_T$ takes the form
\[
V_T - \btheta(P_Z) = \sum_{c=1}^m \binom{m}{c} V_{Tc},
\]
where
\[
V_{Tc} = \int_{\mathcal{V}^c} \bphi_c(\bz_1, \ldots, \bz_c) \prod_{j=1}^c \left( P_T(\dd \bz_j) - P_Z(\dd \bz_j) \right), \qquad c=1,\ldots,m.
\]
In the definition of $V_{Tc}$ we used the empirical distribution
\[
P_T(\bZ) = \frac{1}{T} \left( \delta_{\bZ_1}(\bZ) + \ldots + \delta_{\bZ_T}(\bZ) \right),
\]
where we used Dirac measures $\delta_{\bZt}$, and we defined
\[
\bphi_c(\bz_1, \ldots, \bz_c) = \int_{\mathcal{V}^{m-c}} \bphi(\bz_1, \ldots, \bz_m) \prod_{j=c+1}^m P_Z(\dd \bz_j).
\]
Intuitively, $V_{Tc}$ measures the difference between the sample realization of a consistent estimator and its population counterpart. It converges in probability to the zero vector. The rate of convergence is given by the lemma below, for which we need an additional technical assumption.

\begin{assumption} \label{ass:integrable_phi}
    $\int_{\mathcal{V}^m} \norm{\bphi(\bz_1, \ldots, \bz_m)}^2 \prod_{j=1}^m P_Z(\dd \bz_j) < \infty.$
\end{assumption}

\begin{lemma} \label{lem:Sen}
    Consider a stationary sequence $\bZt \in \mathcal{V}$ obeying the same conditions as in Assumption \ref{ass:psi_mixing_y}. Introduce the $V$-statistics as done above and let Assumption \ref{ass:integrable_phi} hold. Then $\Ex{\norm{V_{Tc}}^2} = \bigO{T^{-c}}$ for $c=1,\ldots,m.$    
\end{lemma}
For a proof of Lemma \ref{lem:Sen}, we refer to Lemma 3.3 in \citet{Sen1972}.

\subsection{Asymptotics of optimal denoising}
The asymptotic behavior of optimal denoising, as stated in Theorem \ref{thm:asymptotics}, is expressed in terms of a $\bigOp{\cdot}$-property. Before proving Theorem \ref{thm:asymptotics}, we first define this notation and relate it to the asymptotic result of expected values from Lemma \ref{lem:Sen}. For a sequence of scalar-valued random variables $\mathcal{W}_T,$ where $T=1,2,\ldots,$ and a sequence of nonzero scalars $a_T,$ we write $\mathcal{W}_T = \bigOp{a_T}$ if for any $\varepsilon>0$ there exist $M>0$ and $T_0 \in \mathbb{N}_+$ such that
\[
\prob{\left| \frac{\mathcal{W}_T}{a_T} \right| \geq M} \leq \varepsilon
\]
for any $T \geq T_0.$ In that case, we say that the sequence of random variables $ \frac{\mathcal{W}_T}{a_T}$ is stochastically bounded.

\begin{lemma} \label{lem:Markov}
    For a sequence of scalar-valued random variables $\mathcal{W}_T,$ if $\Ex{\mathcal{W}_T^2} = \bigO{T^{-c}}$ for some $c>0,$ then $\mathcal{W}_T = \bigOp{T^{-c/2}}.$
\end{lemma}

\begin{proof}
    We use Markov's inequality. Choose an arbitrary $\varepsilon>0.$ Because of the assumption $\Ex{\mathcal{W}_T^2} = \bigO{T^{-c}}$ there exist $M'>0$ and $T_0 \in \mathbb{N}_+$ such that $\Ex{\mathcal{W}_T^2} \leq M' T^{-c}$ for all $T \geq T_0.$ Define $M = \sqrt{M'/\varepsilon.}$ We then have
    \begin{align*}
        \prob{\left| \frac{\mathcal{W}_T}{T^{-c/2}} \right| \geq M}
        & = \prob{\mathcal{W}_T^2 \geq M^2 T^{-c}} \\
        & \leq \frac{\Ex{\mathcal{W}_T^2}}{M^2 T^{-c}} \\
        & \leq \frac{M'}{M^2} \\
        & = \varepsilon,
    \end{align*}
    where the first inequality is Markov's inequality. Since $\varepsilon>0$ was chosen arbitrarily, this proves that $\mathcal{W}_T = \bigOp{T^{-c/2}}.$
\end{proof}

We are now in a position to prove Theorem \ref{thm:asymptotics}, which is the main result on the asymptotic behavior of optimal denoising in this paper.

\begin{proof}[Proof of Theorem \ref{thm:asymptotics}.]
    For simplicity, we assume that the dynamic and noise spaces do not have a common subspace ($\mathcal{M} \cap \mathcal{M}_\varepsilon = \emptyset$). Extending the proof to the case $\mathcal{M} \cap \mathcal{M}_\varepsilon \neq \emptyset$ is straightforward. Throughout this proof, we work with $V V^\top$ instead of $U U^\top.$ As argued in Section \ref{sec:estimation}, they are equal, but only for $V$ can we define an estimator. This makes $V$ suitable for asymptotic analysis, contrary to $U.$

    We know from Theorem \ref{thm:remaining_noise} that optimal denoising removes all noise if $\mathcal{M} \cap \mathcal{M}_\varepsilon = \empty.$ One way of seeing this is by means of the following decomposition,
    \begin{align}
    \normsq{\bxtopt - \bxt}
    & = \normsq{(\Popt - I_n) \bxt + \Popt \bepst} \nonumber \\
    & \leq \normsq{(\Popt - I_n) \bxt} + \normsq{ \Popt \bepst} \nonumber \\
    & \leq \normsq{(VV^\top - I_n) \bxt} + \normsq{ V V^\top \Sigmay W_\perp (\Sigmachi)^{-1}} \normsq{W_\perp^\top \bxt} \label{eq:error_decomposition} \\
    & \quad + \normsq{V} \normsq{ V^\top \bepstperp} + \normsq{ V V^\top \Sigmay W_\perp (\Sigmachi)^{-1}} \normsq{W_\perp^\top \bepstpara} \nonumber \\
    & \quad + \normsq{V V^\top \bepstpara - V V^\top \Sigmay W_\perp (\Sigmachi)^{-1} W_\perp^\top \bepstperp} , \nonumber
    \end{align}
    where we used the triangle inequality and the fact that operator norms are submultiplicative. In the last expression, the rightmost norm in each term is exactly zero, making $\normsq{\bxtopt - \bxt} = 0$ (see the proof of Theorem \ref{thm:remaining_noise} for details). At the sample level, given a time series of length $T,$ each matrix in the above display is replaced by its sample estimate $\Vhat, \Sigmayhat, \Wperphat$ and $\Sigmachihat.$ This makes $\normsq{\bxthatopt - \bxt}$ not exactly zero. We are interested in its behavior as $T\to\infty$ at fixed, finite $n, d$ and $d_\varepsilon.$

    We start by focusing on the first term $\normsq{(\Vhat \Vhat^\top - I_n) \bxt}.$ The columns of $\Vhat$ are the orthonormal eigenvectors of $\Khat$ associated with non-zero eigenvalues. The asymptotics of $\Khat$ have been derived by \citet{Lam2011}. Here, we outline a different derivation based on V-statistics, completely analogous to the derivation for functional time series in \citet{Bathia2010}. Because $\normsq{\Khat - K} \leq \sum_{k=1}^{k_0} c_k^2 \normsq{\Sigmaylaggedhat{k} \Sigmaylaggedhat{k}^\top - \Sigmaylagged{k} \Sigmaylagged{k}^\top},$ we focus on the asymptotics of the $k$-th term in this sum. The statistic
    \[
    \Sigmaylaggedtilde{k} = \frac{1}{T - k} \sum_{t=1}^{T-k} (\by_{t+k} - \bmu) (\byt - \bmu)^\top,
    \]
    where $\bmu = \Ex{\byt},$ is unbiased for $\Sigmaylagged{k},$ in contrast to $\Sigmaylaggedhat{k}.$ Using that $\Sigmaylaggedhat{k} = \Sigmaylaggedtilde{k} - (\bmu - \bybar)(\bmu - \bybar)^\top + \bigOp{T^{-1}}$ for fixed $k,$ we obtain
    \begin{align*}
        \normsq{\Sigmaylaggedhat{k} \Sigmaylaggedhat{k}^\top - \Sigmaylagged{k} \Sigmaylagged{k}^\top} 
        & \leq \normsq{\Sigmaylaggedtilde{k} \Sigmaylaggedtilde{k}^\top - \Sigmaylagged{k} \Sigmaylagged{k}^\top} \\
        & \quad + 2 \normsq{\Sigmaylaggedtilde{k}} \norm{\bmu - \bybar}^4 + \norm{\bmu - \bybar}^8,
    \end{align*}
    where we used the triangle inequality, the fact that operator norms are sub-multiplicative and that the matrix norm of an outer product equals the product of the respective vector norms. The central limit theorem for stationary processes \citep{Anderson1971} implies that $\norm{\bmu - \bybar} = \bigOp{T^{-1/2}},$ making, as we will see, the second and third term in the above expression asymptotically negligible.

    If we define the function $\rho: \mathbb{R}^{n\times n} \times \mathbb{R}^{n\times n} \to \mathbb{R}^{n\times n}$ as $\rho(A, B) = \frac{1}{2}(AB^\top + BA^\top),$ we have that
    \[
    \Sigmaylaggedhat{k} \Sigmaylaggedhat{k}^\top = \frac{1}{(T-k)^2} \sum_{t=1}^{T-k} \sum_{t'=1}^{T-k} \rho (z_t(k), z_{t'}(k)),
    \]
    where $z_t(k) = (\by_{t+k} - \bmu) (\byt - \bmu)^\top.$ This has the form of a V-statistic with $m=2.$ If we define $V_T = \Sigmaylaggedhat{k} \Sigmaylaggedhat{k}^\top,$ then $\theta(P_Z) = \Sigmaylagged{k} \Sigmaylagged{k}^\top$ is its expectation with respect to the product measure as in Equation \eqref{eq:theta_V_statistic}. Since $V_T - \theta(P_Z) = 2 V_{T1} + V_{T2},$ by means of Lemma \ref{lem:Sen} we find that 
    \[
    \Ex{\normsq{\Sigmaylaggedtilde{k} \Sigmaylaggedtilde{k}^\top - \Sigmaylagged{k} \Sigmaylagged{k}^\top}} = \bigO{T^{-1}}.
    \]
    By Lemma \ref{lem:Markov}, we thus find that $\normsq{\Sigmaylaggedhat{k} \Sigmaylaggedhat{k}^\top - \Sigmaylagged{k} \Sigmaylagged{k}^\top} = \bigOp{T^{-1}}$ and therefore $\normsq{\Khat - K} = \bigOp{T^{-1}}.$ 
    
    Assuming that the non-zero eigenvalues of $K$ are all distinct, the distance between the respective eigenvectors is bounded by $\norm{\Khat - K}$ (see theorem 2.7 in \citet{Horvath2012}). As a consequence, $\normsq{\Vhat - V} = \bigOp{T^{-1}}.$ Defining $\delta V = \Vhat - V,$ we get by expanding in $\delta V$
    \begin{align}
        \normsq{(\Vhat \Vhat^\top - I_n) \bxt}
        & = \normsq{(V V^\top - I_n) \bxt + (V \delta V^\top + \delta V V^\top + \delta V \delta V^\top) \bxt} \label{eq:asymptotic_expansion} \\
        & \leq \normsq{(V V^\top - I_n) \bxt} + 2 \normsq{V} \normsq{\delta V} \normsq{\bxt} + \norm{\delta V}^4 \normsq{\bxt} \nonumber \\
        & = 0 + \bigOp{T^{-1}} + \bigOp{T^{-2}} \nonumber \\
        & = \bigOp{T^{-1}} . \nonumber
    \end{align}
    Going back to Equation \eqref{eq:error_decomposition}, with analogous arguments we obtain for the sample version of the third term that $\normsq{\Vhat} \normsq{ \Vhat^\top \bepstperp} = \bigOp{T^{-1}}$ as well.

    For the other terms in Equation \eqref{eq:error_decomposition}, we need to know the asymptotic behavior of $\Sigmayhat, \Wperphat$ and $\Sigmachihat.$ It is well-known that $\normsq{\Sigmayhat} = \bigOp{T^{-1}}$ \citep{Brockwell1991}. Recall that the columns of $\Wperphat$ are the orthonormal eigenvectors of $\Sigmaepsperphat = (I_n - \Vhat \Vhat^\top) \Sigmayhat (I_n - \Vhat \Vhat^\top)$ associated with the $d_\perp$ largest eigenvalues. Analogously to Equation \eqref{eq:asymptotic_expansion}, we can expand $\Sigmaepsperphat$ in terms of $\delta V$ and $\Sigmayhat - \Sigmay,$ use their asymptotics and the triangle inequality, to find that $\normsq{\Sigmaepsperphat - \Sigmaepsperp} = \bigOp{T^{-1}}.$ We again invoke Theorem 2.7 of \citet{Horvath2012} and find $\normsq{\Wperphat} = \bigOp{T^{-1}}.$ The matrix $\Sigmachihat$ is diagonal, with on its diagonal the $d_\perp$ largest eigenvalues of $\Sigmaepsperphat.$ The Bauer-Fike theorem implies that the gap between the respective eigenvalues of $\Sigmaepsperphat$ and $\Sigmaepsperp$ is bounded by $\norm{\Sigmaepsperphat - \Sigmaepsperp}$ and hence $\normsq{\Sigmachihat - \Sigmachi} = \bigOp{T^{-1}}.$ Proceeding as in Equation \eqref{eq:asymptotic_expansion}, by expanding the remaining terms in the sample version of Equation \eqref{eq:error_decomposition} we find that all terms have the same leading-order asymptotic behavior of $\bigOp{T^{-1}}.$ This concludes the proof.    
\end{proof}

\section{Simulation - details} \label{app:simulation_details}
The procedure that randomly generates VAR($p$) processes as defined in Equation \eqref{eq:VAR_process} is given by Algorithm \ref{alg:VAR_process_generation}. For technical details about VAR($p$) processes, see, e.g., \cite{Lutkepohl2005}. We require that the process is stationary by imposing that the largest modulus of the roots of the characteristic polynomial of this VAR($p$) process $\text{r}_{\text{max}}(A_1, \ldots, A_s)$ is below $r_2<1,$ where $r_2$ is set by the practitioner. We also ensure that $\Sigmaxi$ is diagonal and that $\Tr{\Sigmaxi}=1.$ Furthermore, we require that $\Tr{\Sigma_{\mathrm{\bf e}}}$ is  within a specified range $[s_1,s_2],$ with $s_1, s_2 \in (0,1)$, which facilitates comparison of the asymptotic behavior across VAR($p$) processes.

The algorithm consists of four steps. First, a positive definite $\Sigma_{\mathrm{\bf e}}$ is randomly sampled. Then matrices $A_1, A_2, \ldots, A_p$ are randomly sampled and rescaled such that they constitute a stationary VAR($p$) process. Next, the corresponding $\Sigmaxi$ is diagonalized and $\Sigma_{\mathrm{\bf e}}$ is rescaled such that $\Tr{\Sigmaxi}=1.$ If the resulting $\Tr{\Sigma_{\mathrm{\bf e}}} \notin [s_1,s_2],$ the whole procedure is repeated. Otherwise, the procedure is terminated and the output is a VAR($p$) process defined by $\Sigma_{\mathrm{\bf e}}$ and $A_s$ for $s=1,2,\ldots,p.$

Throughout the simulations in this paper, unless stated otherwise, we use $r_1 = 0.01,$ $r_2 = 0.98,$ $s_1 = 0.05,$ and $s_2 = 0.25.$ The most important parameters here are $r_2$ and $s_2.$ Obviously, $r_2 < 1$ to ensure stationarity. If we choose $r_2$ too small, Algorithm \ref{alg:VAR_process_generation} has difficulty finding a VAR($p$) process satisfying all conditions when both $d$ and $p$ are large (e.g., $d=6,$ $p=5$). If we choose $s_2$ too small, we encounter a similar problem. If we choose $s_2$ too large, the convergence of the denoising procedures as a function of the time series length $T$ is slower.

\begin{algorithm}
\caption{Random generation of VAR($p$) process.}\label{alg:VAR_process_generation}
\begin{algorithmic}[1]
\Require $d \geq 1;$ $p \geq 1;$ $r_1, r_2, s_1, s_2 \in (0,1)$ such that $r_1 < r_2$ and $s_1 < s_2.$
\While{$\Tr{\Sigma_{\mathrm{\bf e}}} \notin [s_1,s_2]$}

    \State 
    \While{$\Sigma_{\mathrm{\bf e}}$ is not positive definite} \Comment{Sample $\Sigma_{\mathrm{\bf e}}.$}
        \State $\Sigma_{\mathrm{\bf e}} \gets \frac{1}{10} M M^\top,$ where $M \thicksim \mathcal{N} ({\bf 0}, I_d)$ is randomly sampled.
    \EndWhile
    
    \State
    \State Sample $A_s \thicksim \mathcal{N} ({\bf 0}, I_d)$ for $s=1,2,\ldots,p.$ \Comment{Sample $A_s.$}

    \State
    \State $b \gets (\text{r}_{\text{max}}(A_1, \ldots, A_s) > r_2)$ \Comment{Find stationary process.}
    \State $f \gets 2$
    \While{$\text{r}_{\text{max}}(A_1, \ldots, A_s) \notin [r_1, r_2]$}
        \If{$\text{r}_{\text{max}}(A_1, \ldots, A_s) > r_2$}
            \If{$b=0$}
                \State $b \gets 1$
                \State $f \gets f + 1$
            \EndIf
            \State $A_s \gets (1 - 1/f)A_s$ for $s=1,2,\ldots,p.$
        \Else
            \If{$b=1$}
                \State $b \gets 0$
                \State $f \gets f + 1$
            \EndIf
            \State $A_s \gets (1 + 1/f)A_s$ for $s=1,2,\ldots,p.$
        \EndIf
    \EndWhile
    
    \State
    \State Find $S$ that diagonalizes $\Sigmaxi.$ \Comment{Diagonalize $\Sigmaxi.$}
    \State $A_s \gets S^\top A_s S$ for $s=1,2,\ldots,p.$
    \State $\Sigma_{\mathrm{\bf e}} \gets S^\top \Sigma_{\mathrm{\bf e}} S.$
    
    \State
    \State Rescale $\Sigma_{\mathrm{\bf e}} \gets \Sigma_{\mathrm{\bf e}} / \Tr{\Sigmaxi}$. \Comment{Impose $\Tr{\Sigmaxi}=1.$}
    \State
\EndWhile
\State \Return $\Sigma_{\mathrm{\bf e}}$ and $A_s$ for $s=1,2,\ldots,p.$
\end{algorithmic}
\end{algorithm}

\newpage
\section{Additional simulation results} \label{app:additional_simulation_results}

\begin{figure}[H]
    \centering
    \includegraphics[scale=0.6]{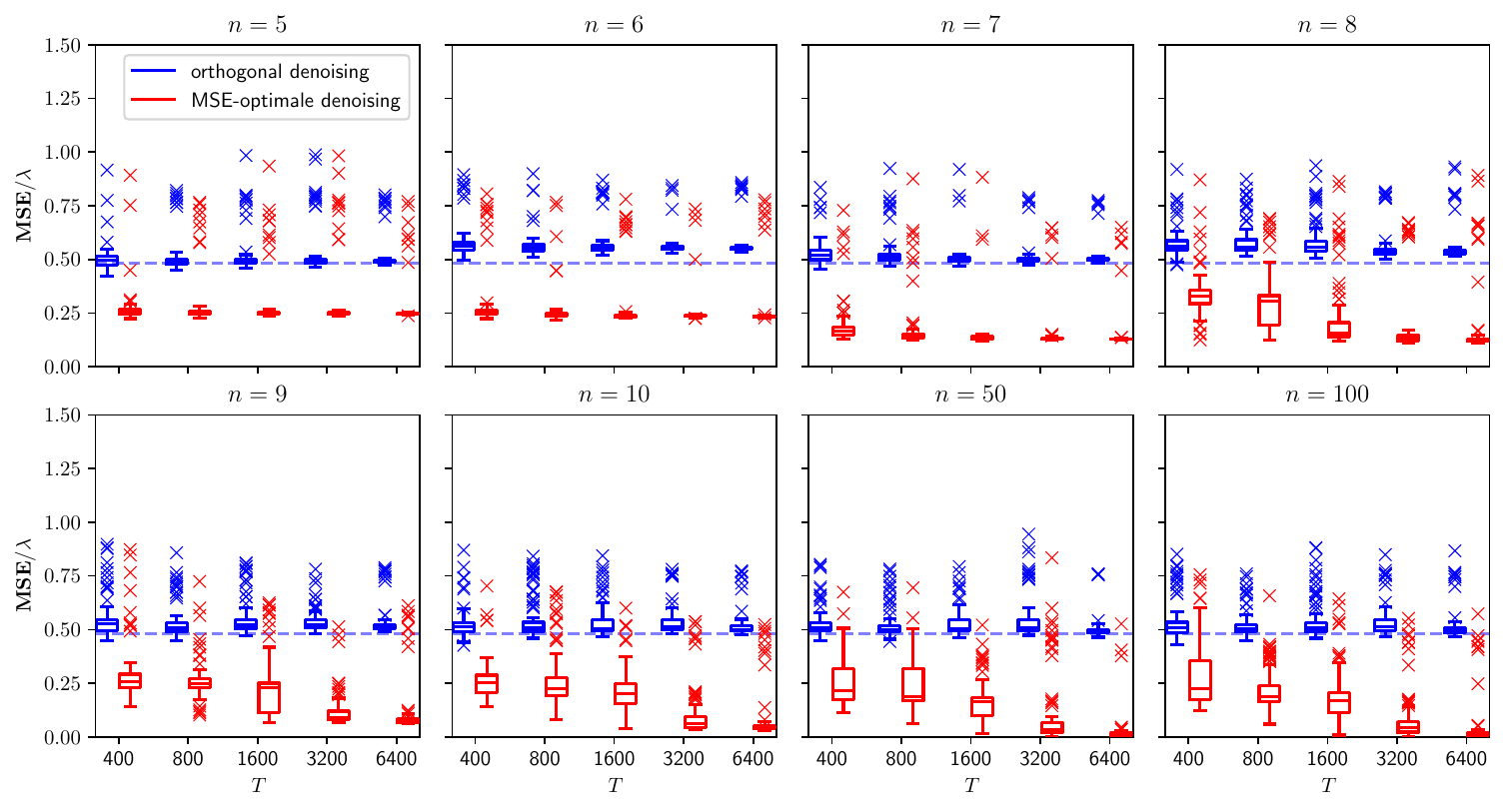}
    \caption{(asymptotic) denoising performance, measured as $\mse/\lambda,$ as a function of the dimension $n$ of the observed time series $\byt.$ Each boxplot consists of 100 independent simulations.}
    \label{fig:dependence_n}
\end{figure}

\begin{figure}[H]
    \centering
    \includegraphics[scale=0.6]{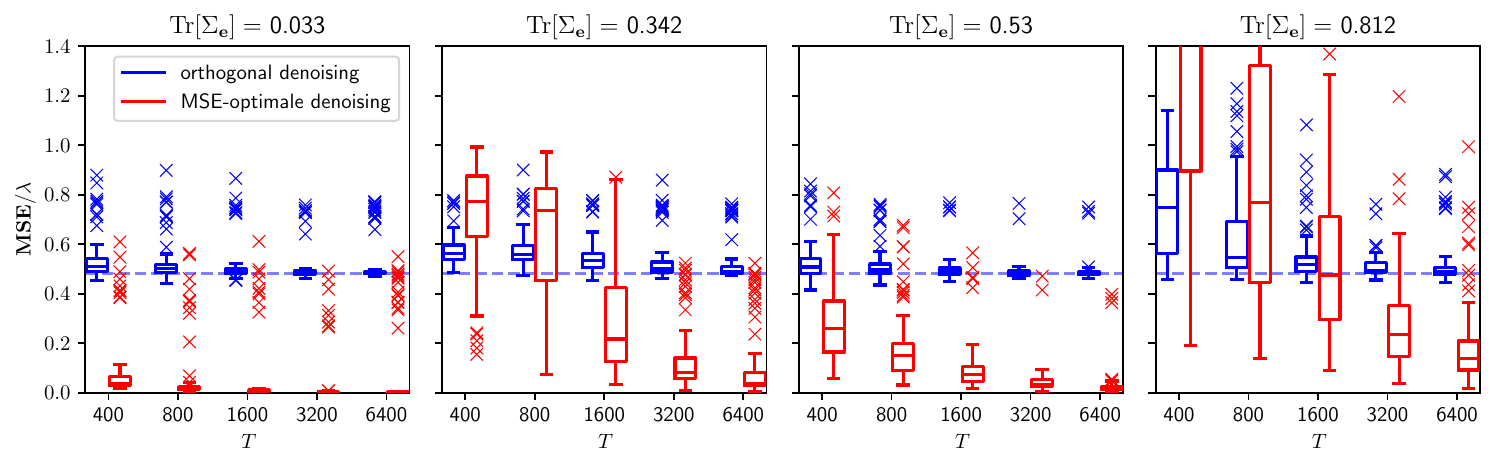}
    \caption{(asymptotic) denoising performance, measured as $\mse/\lambda,$ as a function of the total variance of the VAR($p$) process, $\Tr{\Sigma_{\rm \bf e}}.$  Each boxplot consists of 100 independent simulations.}
    \label{fig:dependence_trace_sigma_e}
\end{figure}

\begin{figure}
    \centering
    \includegraphics[scale=0.6]{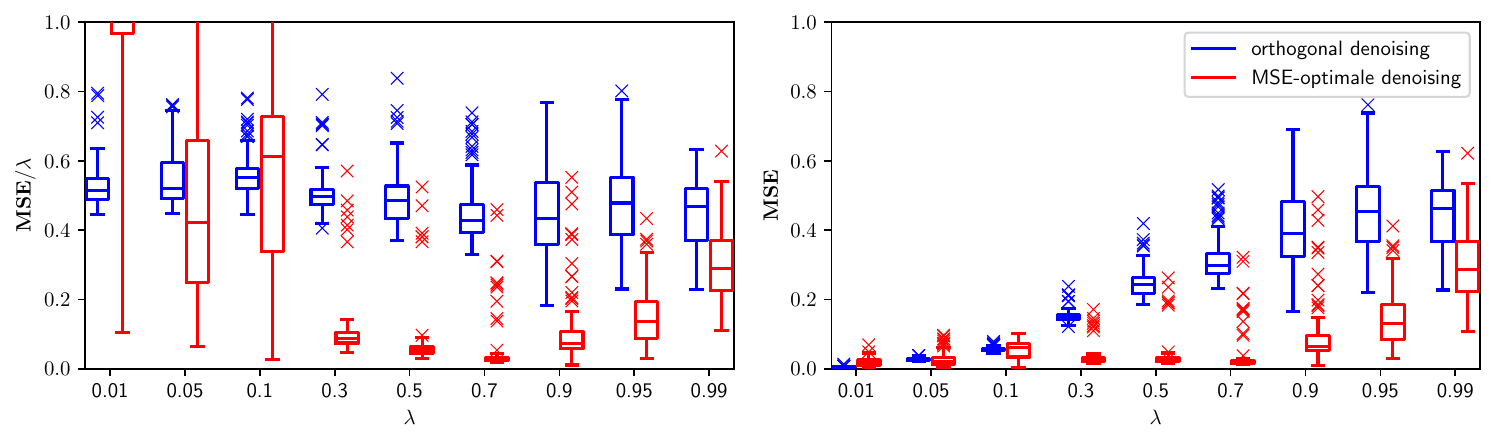}
    \caption{denoising performance for time series of (default) length $T=400,$ measured as $\mse/\lambda$ (left panel) and as absolute $\mse$ (right panel), as a function of the noise fraction $\lambda$ of the observed time series $\{\byt\}$. Each boxplot consists of 100 independent simulations.}
    \label{fig:dependence_lambda}
\end{figure}

\begin{figure}
    \centering
    \includegraphics[scale=0.6]{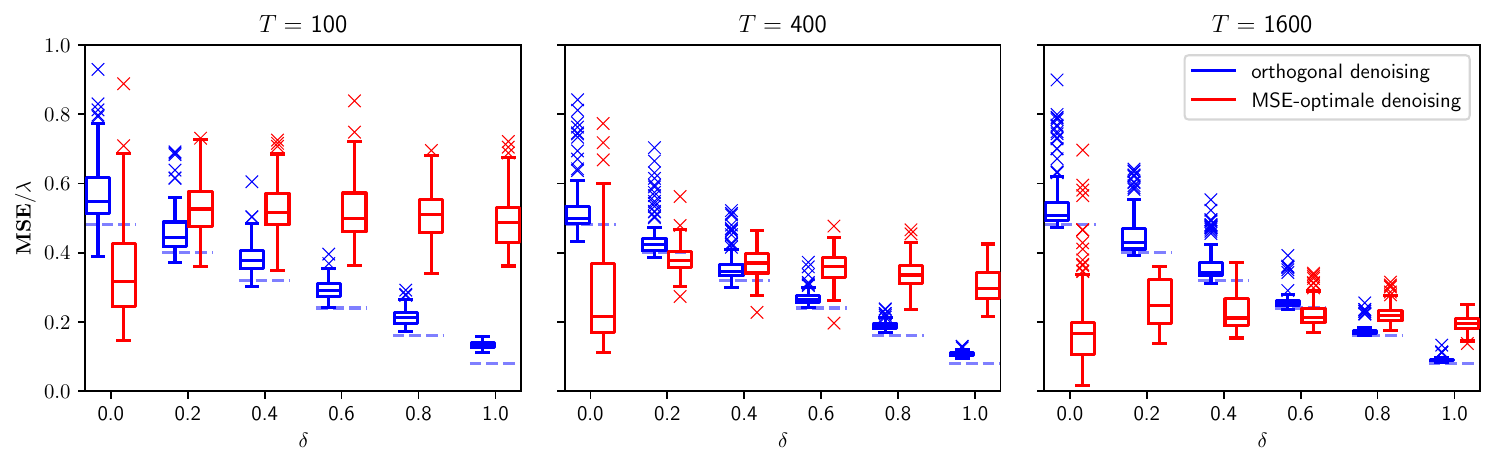}
    \caption{(asymptotic) denoising performance, measured as $\mse/\lambda,$ as a function of the fraction $\delta$ of unstructured noise $\bepstnull$ within the total noise $\bepst + \bepstnull,$ while keeping the fraction of the latter within the observed time series $\{\byt\}$ constant at the (default) value $\lambda=0.2.$ Each boxplot consists of 100 independent simulations.}
    \label{fig:dependence_delta}
\end{figure}

\begin{figure}
    \centering
    \includegraphics[scale=0.6]{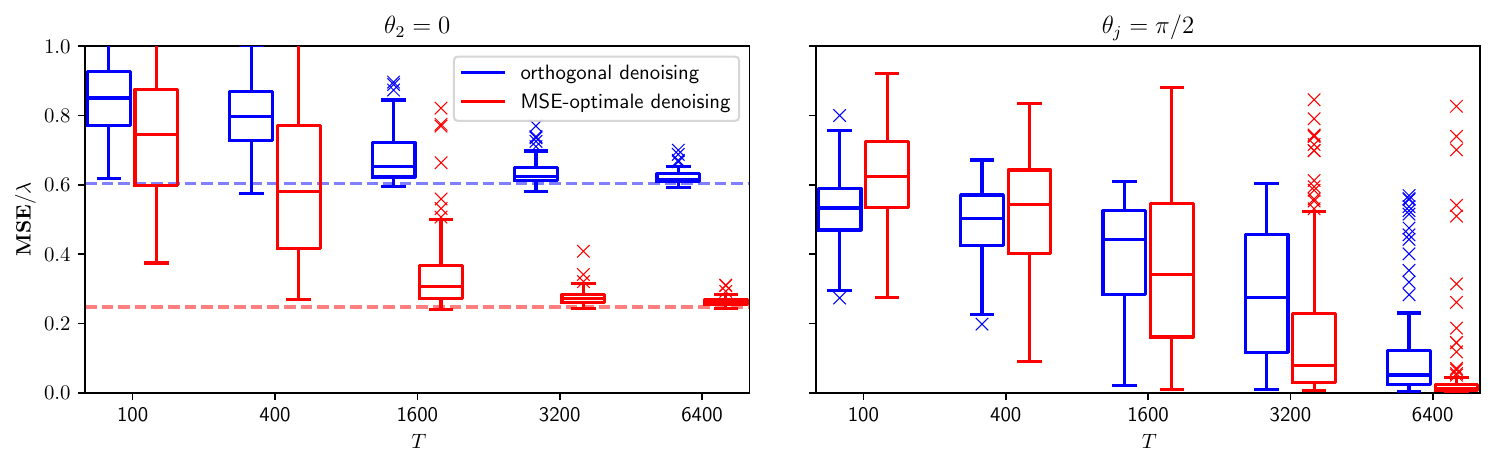}
    \caption{(asymptotic) denoising performance, measured as $\mse/\lambda,$ as a function of the orientation of the noise space with respect to the dynamic space. In the left panel the second noise component is parallel ($\theta_2=0$) to the dynamic space. In the right panel the full noise space is in the orthogonal complement of the dynamic space ($\theta_j=\pi/2$).}
    \label{fig:dependence_thetas}
\end{figure}


\newpage
\bibliographystyle{apalike}
\bibliography{refs.bib}

@book{Anderson1971,
  title = {The {{Statistical Analysis}} of {{Time Series}}},
  author = {Anderson, Theodore W.},
  year = {1971},
  series = {Wiley {{Series}} in {{Probability}} and {{Statistics}}},
  edition = {1st edition},
  publisher = {John Wiley},
  address = {New York},
  isbn = {978-0-471-04745-2}
}

@article{Bai2002,
  title = {Determining the {{Number}} of {{Factors}} in {{Approximate Factor Models}}},
  author = {Bai, Jushan and Ng, Serena},
  year = {2002},
  journal = {Econometrica},
  volume = {70},
  number = {1},
  pages = {191--221},
  issn = {1468-0262},
  doi = {10.1111/1468-0262.00273},
  copyright = {The Econometric Society 2001}
}

@article{Bai2007,
  title = {Determining the {{Number}} of {{Primitive Shocks}} in {{Factor Models}}},
  author = {Bai, Jushan and Ng, Serena},
  year = {2007},
  journal = {Journal of Business \& Economic Statistics},
  volume = {25},
  number = {1},
  pages = {52--60},
  publisher = {ASA Website},
  issn = {0735-0015},
  doi = {10.1198/073500106000000413}
}

@article{Bathia2010,
  title = {Identifying the Finite Dimensionality of Curve Time Series},
  author = {Bathia, Neil and Yao, Qiwei and Ziegelmann, Flavio},
  year = {2010},
  journal = {The Annals of Statistics},
  volume = {38},
  number = {6},
  eprint = {1211.2522},
  primaryclass = {math},
  issn = {0090-5364},
  doi = {10.1214/10-AOS819},
  archiveprefix = {arXiv}
}

@book{Brockwell1991,
  title = {Time {{Series}}: {{Theory}} and {{Methods}}},
  shorttitle = {Time {{Series}}},
  author = {Brockwell, Peter J. and Davis, Richard A.},
  year = {1991},
  series = {Springer {{Series}} in {{Statistics}}},
  publisher = {Springer},
  address = {New York, NY},
  doi = {10.1007/978-1-4419-0320-4},
  copyright = {http://www.springer.com/tdm},
  isbn = {978-1-4419-0319-8 978-1-4419-0320-4}
}

@article{Candes2007,
  title = {The {{Dantzig}} Selector: {{Statistical}} Estimation When p Is Much Larger than n},
  shorttitle = {The {{Dantzig}} Selector},
  author = {Candes, Emmanuel and Tao, Terence},
  year = {2007},
  journal = {The Annals of Statistics},
  volume = {35},
  number = {6},
  issn = {0090-5364},
  doi = {10.1214/009053606000001523}
}

@article{Chen2022,
  title = {Functional {{Linear Regression}}: {{Dependence}} and {{Error Contamination}}},
  shorttitle = {Functional {{Linear Regression}}},
  author = {Chen, Cheng and Guo, Shaojun and Qiao, Xinghao},
  year = {2022},
  journal = {Journal of Business \& Economic Statistics},
  volume = {40},
  number = {1},
  pages = {444--457},
  issn = {0735-0015, 1537-2707},
  doi = {10.1080/07350015.2020.1832503}
}

@article{Donoho1994,
  title = {Ideal Spatial Adaptation by Wavelet Shrinkage},
  author = {Donoho, David L and Johnstone, Iain M},
  year = {1994},
  journal = {Biometrika},
  volume = {81},
  number = {3},
  pages = {425--455},
  issn = {0006-3444},
  doi = {10.1093/biomet/81.3.425}
}

@article{Donoho1995,
  title = {Adapting to {{Unknown Smoothness}} via {{Wavelet Shrinkage}}},
  author = {Donoho, David L. and Johnstone, Iain M.},
  year = {1995},
  journal = {Journal of the American Statistical Association},
  volume = {90},
  number = {432},
  pages = {1200--1224},
  publisher = {Taylor \& Francis},
  issn = {0162-1459},
  doi = {10.1080/01621459.1995.10476626}
}

@article{Fan2011,
  title = {High-Dimensional Covariance Matrix Estimation in Approximate Factor Models},
  author = {Fan, Jianqing and Liao, Yuan and Mincheva, Martina},
  year = {2011},
  journal = {The Annals of Statistics},
  volume = {39},
  number = {6},
  pages = {3320--3356},
  publisher = {Institute of Mathematical Statistics},
  issn = {0090-5364, 2168-8966},
  doi = {10.1214/11-AOS944}
}

@article{Fan2013,
  title = {Large {{Covariance Estimation}} by {{Thresholding Principal Orthogonal Complements}}},
  author = {Fan, Jianqing and Liao, Yuan and Mincheva, Martina},
  year = {2013},
  journal = {Journal of the Royal Statistical Society Series B: Statistical Methodology},
  volume = {75},
  number = {4},
  pages = {603--680},
  issn = {1369-7412},
  doi = {10.1111/rssb.12016}
}

@article{Forni2000,
  title = {The {{Generalized Dynamic-Factor Model}}: {{Identification}} and {{Estimation}}},
  shorttitle = {The {{Generalized Dynamic-Factor Model}}},
  author = {Forni, Mario and Hallin, Marc and Lippi, Marco and Reichlin, Lucrezia},
  year = {2000},
  journal = {The Review of Economics and Statistics},
  volume = {82},
  number = {4},
  pages = {540--554},
  issn = {0034-6535},
  doi = {10.1162/003465300559037}
}

@article{Forni2005,
  title = {The {{Generalized Dynamic Factor Model}}: {{One-Sided Estimation}} and {{Forecasting}}},
  shorttitle = {The {{Generalized Dynamic Factor Model}}},
  author = {Forni, Mario and Hallin, Marc and Lippi, Marco and Reichlin, Lucrezia},
  year = {2005},
  journal = {Journal of the American Statistical Association},
  volume = {100},
  number = {471},
  pages = {830--840},
  publisher = {Taylor \& Francis},
  issn = {0162-1459},
  doi = {10.1198/016214504000002050}
}

@article{Hallin2007,
  title = {Determining the {{Number}} of {{Factors}} in the {{General Dynamic Factor Model}}},
  author = {Hallin, Marc and Li{\v s}ka, Roman},
  year = {2007},
  journal = {Journal of the American Statistical Association},
  volume = {102},
  number = {478},
  pages = {603--617},
  publisher = {ASA Website},
  issn = {0162-1459},
  doi = {10.1198/016214506000001275}
}

@book{Horvath2012,
  title = {Inference for {{Functional Data}} with {{Applications}}},
  author = {Horv{\'a}th, Lajos and Kokoszka, Piotr},
  year = {2012},
  series = {Springer {{Series}} in {{Statistics}}},
  volume = {200},
  publisher = {Springer},
  address = {New York, NY},
  doi = {10.1007/978-1-4614-3655-3},
  copyright = {https://www.springernature.com/gp/researchers/text-and-data-mining},
  isbn = {978-1-4614-3654-6 978-1-4614-3655-3}
}

@article{Hyndman2007,
  title = {Robust Forecasting of Mortality and Fertility Rates: {{A}} Functional Data Approach},
  shorttitle = {Robust Forecasting of Mortality and Fertility Rates},
  author = {Hyndman, Rob J. and Shahid Ullah, {\relax Md}.},
  year = {2007},
  journal = {Computational Statistics \& Data Analysis},
  volume = {51},
  number = {10},
  pages = {4942--4956},
  issn = {01679473},
  doi = {10.1016/j.csda.2006.07.028},
  copyright = {https://www.elsevier.com/tdm/userlicense/1.0/}
}

@article{Johnstone2009,
  title = {On {{Consistency}} and {{Sparsity}} for {{Principal Components Analysis}} in {{High Dimensions}}},
  author = {Johnstone, Iain M. and Lu, Arthur Yu},
  year = {2009},
  journal = {Journal of the American Statistical Association},
  volume = {104},
  number = {486},
  pages = {682--693},
  publisher = {Taylor \& Francis},
  issn = {0162-1459},
  doi = {10.1198/jasa.2009.0121}
}

@article{Lam2011,
  title = {Estimation of Latent Factors for High-Dimensional Time Series},
  author = {Lam, C. and Yao, Q. and Bathia, N.},
  year = {2011},
  journal = {Biometrika},
  volume = {98},
  number = {4},
  pages = {901--918},
  issn = {0006-3444, 1464-3510},
  doi = {10.1093/biomet/asr048}
}

@article{Lam2012,
  title = {Factor Modeling for High-Dimensional Time Series: {{Inference}} for the Number of Factors},
  shorttitle = {Factor Modeling for High-Dimensional Time Series},
  author = {Lam, Clifford and Yao, Qiwei},
  year = {2012},
  journal = {The Annals of Statistics},
  volume = {40},
  number = {2},
  issn = {0090-5364},
  doi = {10.1214/12-AOS970}
}

@article{Ledoit2004,
  title = {A Well-Conditioned Estimator for Large-Dimensional Covariance Matrices},
  author = {Ledoit, Olivier and Wolf, Michael},
  year = {2004},
  journal = {Journal of Multivariate Analysis},
  volume = {88},
  number = {2},
  pages = {365--411},
  issn = {0047-259X},
  doi = {10.1016/S0047-259X(03)00096-4}
}

@book{Lutkepohl2005,
  title = {New {{Introduction}} to {{Multiple Time Series Analysis}}},
  author = {L{\"u}tkepohl, Helmut},
  year = {2005},
  publisher = {Springer},
  address = {Berlin, Heidelberg},
  doi = {10.1007/978-3-540-27752-1},
  copyright = {http://www.springer.com/tdm},
  isbn = {978-3-540-40172-8 978-3-540-27752-1}
}

@misc{Otto2024,
  title = {Approximate {{Factor Models}} for {{Functional Time Series}}},
  author = {Otto, Sven and Salish, Nazarii},
  year = {2024},
  number = {arXiv:2201.02532},
  eprint = {2201.02532},
  primaryclass = {econ},
  publisher = {arXiv},
  doi = {10.48550/arXiv.2201.02532},
  archiveprefix = {arXiv}
}

@misc{Otto2025,
  title = {Functional {{Factor Regression}} with an {{Application}} to {{Electricity Price Curve Modeling}}},
  author = {Otto, Sven and Winter, Luis},
  year = {2025},
  number = {arXiv:2503.12611},
  eprint = {2503.12611},
  primaryclass = {econ},
  publisher = {arXiv},
  doi = {10.48550/arXiv.2503.12611},
  archiveprefix = {arXiv}
}

@article{Pan2008,
  title = {Modelling Multiple Time Series via Common Factors},
  author = {Pan, J. and Yao, Q.},
  year = {2008},
  journal = {Biometrika},
  volume = {95},
  number = {2},
  pages = {365--379},
  issn = {0006-3444, 1464-3510},
  doi = {10.1093/biomet/asn009}
}

@article{Sen1972,
  title = {Limiting Behavior of Regular Functionals of Empirical Distributions for Stationary *-Mixing Processes},
  author = {Sen, Pranab Kumar},
  year = {1972},
  journal = {Zeitschrift f{\"u}r Wahrscheinlichkeitstheorie und Verwandte Gebiete},
  volume = {25},
  number = {1},
  pages = {71--82},
  issn = {0044-3719, 1432-2064},
  doi = {10.1007/BF00533337},
  copyright = {http://www.springer.com/tdm}
}

@article{Stock2002,
  title = {Macroeconomic {{Forecasting Using Diffusion Indexes}}},
  author = {Stock, James H and Watson, Mark W},
  year = {2002},
  journal = {Journal of Business \& Economic Statistics},
  volume = {20},
  number = {2},
  pages = {147--162},
  publisher = {Taylor \& Francis},
  issn = {0735-0015},
  doi = {10.1198/073500102317351921}
}

@article{Stock2002a,
  title = {Forecasting {{Using Principal Components From}} a {{Large Number}} of {{Predictors}}},
  author = {Stock, James H and Watson, Mark W},
  year = {2002},
  journal = {Journal of the American Statistical Association},
  volume = {97},
  number = {460},
  pages = {1167--1179},
  publisher = {Taylor \& Francis},
  issn = {0162-1459},
  doi = {10.1198/016214502388618960}
}

@book{Vaart_1998, place={Cambridge}, 
    series={Cambridge Series in Statistical and Probabilistic Mathematics}, title={Asymptotic Statistics}, 
    publisher={Cambridge University Press, Cambridge}, 
    author={Van der Vaart, A. W. }, 
    year={1998}, 
    collection={Cambridge Series in Statistical and Probabilistic Mathematics}
}

@article{Wu2018,
  title = {Eigenvalue Difference Test for the Number of Common Factors in the Approximate Factor Models},
  author = {Wu, Jianhong},
  year = {2018},
  journal = {Economics Letters},
  volume = {169},
  pages = {63--67},
  issn = {01651765},
  doi = {10.1016/j.econlet.2018.05.009}
}

@article{Xia2018,
  title = {Determining the Number of Factors for High-Dimensional Time Series},
  author = {Xia, Qiang and Liang, Rubing and Wu, Jianhong and Wong, Heung},
  year = {2018},
  journal = {Statistics and Its Interface},
  volume = {11},
  number = {2},
  pages = {307--316},
  issn = {19387989, 19387997},
  doi = {10.4310/SII.2018.v11.n2.a8}
}

\end{document}